\documentclass[fleqn,11pt,twoside]{article}

\usepackage{amsthm,amsthm,amssymb, color, xcolor,epsfig, graphics, subfigure}

\usepackage{amsmath, graphicx, latexsym, lscape}

\usepackage{tikz}

\makeatletter
\newcommand{\copyrightnote}[2]{{\renewcommand{\thefootnote}{}
 \footnotetext{\small\it
\begin{flushleft}
 \copyright \ #1   #2
\end{flushleft}}}}

\newcommand{\Name}[1]{\begin{flushleft}
                       \LARGE \bf #1
                       \end{flushleft}\vspace{-3mm}}

\newcommand{\Author}[1]{\begin{flushleft}
                       \it #1 \end{flushleft}}

\newcommand{\Address}[1]{\begin{flushleft}
                       \it #1 \end{flushleft}}

\newcommand{\Date}[1]{\begin{flushleft}
                      \small  \it #1 \end{flushleft}}

\newcommand{\evenhead}{Author \ name}
\newcommand{\oddhead}{Article \ name}

\renewcommand{\@evenhead}{
\hspace*{-3pt}\raisebox{-15pt}[\headheight][0pt]{\vbox{\hbox to \textwidth
{\thepage \hfil \evenhead}\vskip4pt \hrule}}}
\renewcommand{\@oddhead}{
\hspace*{-3pt}\raisebox{-15pt}[\headheight][0pt]{\vbox{\hbox to \textwidth
{\oddhead \hfil \thepage}\vskip4pt\hrule}}}
\renewcommand{\@evenfoot}{}
\renewcommand{\@oddfoot}{}

\long\def\@makecaption#1#2{%
  \vskip\abovecaptionskip
  \sbox\@tempboxa{\small \textbf{#1.}\ \ #2}%
  \ifdim \wd\@tempboxa >\hsize
    {\small \textbf{#1.}\ \ #2}\par
  \else
    \global \@minipagefalse
    \hb@xt@\hsize{\hfil\box\@tempboxa\hfil}%
  \fi
  \vskip\belowcaptionskip}

\newcommand{\JNMPnumberwithin}[3][\arabic]{%
  \@ifundefined{c@#2}{\@nocounterr{#2}}{%
    \@ifundefined{c@#3}{\@nocnterr{#3}}{%
      \@addtoreset{#2}{#3}%
      \@xp\xdef\csname the#2\endcsname{%
        \@xp\@nx\csname the#3\endcsname .\@nx#1{#2}}}}%
}

\renewenvironment{proof}[1][\proofname]{\par
  \normalfont
  \topsep6\p@\@plus6\p@ \trivlist
  \item[\hskip\labelsep\textbf{%
    #1\@addpunct{.}}]\ignorespaces
}{%
  \qed\endtrivlist
}

\newcommand{\resetfootnoterule} {
  \renewcommand\footnoterule{%
  \kern-3\p@
  \hrule\@width.4\columnwidth
  \kern2.6\p@}
}

\renewcommand{\footnoterule}{}

\makeatother

\newtheorem{lem}{Lemma}[section]
\newtheorem{definition}{Definition}[section]
\newtheorem{theorem}{Theorem}[section]

\newtheorem{property}{Property}[section]

\allowdisplaybreaks

\def \ccirc{{\hbox{\small{$\circ$}}}}
\def \t#1{\widetilde{#1}}

\begin{document}

\renewcommand{\evenhead}{ {\LARGE\textcolor{blue!10!black!40!green}{{\sf \ \ \ ]ocnmp[}}}\strut\hfill
Da-jun Zhang
}
\renewcommand{\oddhead}{ {\LARGE\textcolor{blue!10!black!40!green}{{\sf ]ocnmp[}}}\ \ \ \ \
Integrability and transformations in the bilinear method
}

\thispagestyle{empty}
\newcommand{\FistPageHead}[3]{
\begin{flushleft}
\raisebox{8mm}[0pt][0pt]
{\footnotesize \sf
\parbox{150mm}{{\textcolor{blue!10!black!40!green}{{\bf Open Communications in Nonlinear Mathematical Physics}}}
\ \ {Special Issue: Hietarinta}, 2026\\[0.1cm]
\strut\hfill
ocnmp:18521
pp #2\hfill {\sc #3}}}\vspace{-13mm}
\end{flushleft}}

\FistPageHead{1}{\pageref{firstpage}--\pageref{lastpage}}{ \ \ }

\strut\hfill

\strut\hfill

\copyrightnote{The authors. Distributed under a Creative Commons Attribution 4.0 International License}

\begin{center}

{\bf {\large A Special OCNMP Issue in Honour of Jarmo Hietarinta}}\\[0.2cm]
{\bf {\large on the Occasion of his 80th Birthday}}
\end{center}

\smallskip

\Name{Integrability and Transformations in the Bilinear Method: An Introduction}

\Author{Da-jun Zhang\footnote{
Email: djzhang@staff.shu.edu.cn}}

\Address{Department  of Mathematics, Shanghai University, Shanghai 200444, China }

\Date{Received June 16, 2026; Accepted July 21, 2026}

\setcounter{equation}{0}

\smallskip

\noindent
{\bf Citation format for this Article:}\newline
Da-jun Zhang,
Integrability and transformations in the bilinear method: An introduction,
{\it Open Commun. Nonlinear Math. Phys.}, Special Issue:\,Hietarinta, ocnmp:18521, \pageref{firstpage}--\pageref{lastpage}, 2026.

\strut\hfill

\noindent
{\bf The permanent Digital Object Identifier (DOI) for this Article:}\newline
{\it 10.46298/ocnmp.18521}
\strut\hfill

\begin{abstract}

\noindent
This is a partial review of  the bilinear method,
focusing on the integrability based on the 3-soliton-solution condition
and the transformations between $\tau$ functions.
\textit{Dedicated to Jarmo Hietarinta's 80th birthday}.

\end{abstract}

\label{firstpage}


\section{Introduction}\label{sec-1}

In 1971, Ryogo Hirota invented the celebrated bilinear method \cite{Hir-1971}.
According his recollection in \cite{Hir-2004} and the report by Daisuke Takahashi \cite{Tak-2016},
he was motivated by the research of Morikazu Toda
who found a 2-soliton solution (2SS) for the Toda lattice in 1968.
He tried to construct a 3-soliton solution (3SS) for this equation but he failed.
The breakthrough came up in April of 1971: He got a 3SS for the sine-Gordon equation
using the superposition formula (known as the Bianchi identity) of its B\"acklund transformation.
Later, he transformed the nonlinear self-dual network equation into a bilinear form
and from which he obtained a 3SS.
Then, he successively obtained the bilinear forms and 3SSs for the Toda lattice,
the Korteweg-de Vries (KdV) equation and the sine-Gordon equation.
In his celebrated paper \cite{Hir-1971} in 1971, the $N$-soliton solution (NSS)
of the bilinear  KdV equation is presented in the determinant of a generalized Cauchy matrix.
Soon after, this form was re-derived by Wadati and Toda \cite{WadT-1972}
from the  Inverse Scattering Transform \cite{GGKM-1967}.
Without the time dependence, this form has been obtained by Kay and Moses \cite{KayM-1956} in 1956,
from the  solution of the Gel'fand-Levitan-Marchenko (GLM) equation \cite{GelL-1951,Mar-1955} in the reflectionless case.
Hirota's solution is nowadays known as the $\tau$ function of the KdV equation,
while, as  a notion, $\tau$ function has played a more general and more important role in the theory of integrable systems, far beyond the solutions of bilinear equations.

In 1974, Hirota introduced the bilinear operator ``$D$'' in \cite{Hir-1974}
and started to study bilinear B\"acklund transformations (BTs).
In the same paper he also explored the relation between  bilinear BTs
and Lax pairs, which connects the BTs and integrability from the bilinear perspective.
In 1980, Hirota reported the fact that any KdV-type bilinear equation always has a 1SS and a 2SS,
no matter whether it is integrable or not \cite{Hir-1980}.
It is natural to realize that having a 3SS should mean ``something'' for a KdV-type bilinear equation.
Hirota came up with a question in \cite{Hir-1980}: Under what conditions a bilinear equation may have a NSS
in certain form (see Sec.\ref{sec-3-2}), which is called the bilinear equation being integrable in Hirota's sense.
In 1987, Jarmo Hietarinta searched for possible KdV-type bilinear equations that admit 3SSs \cite{Hie-1987a}.
Later, he extended the search to the bilinear equations of the modified KdV (mKdV) type and
sine-Gordon type under the conditions of having 3SSs \cite{Hie-1987b,Hie-1987c},
and to the complex bilinear equations (e.g. the nonlinear Schr\"odinger (NLS) type
and the Benjamin-Ono (BO) type) \cite{Hie-1988}, for which the existence of a 2SS may imply integrability.
Hietarinta summarized his search in \cite{Hie-1990}.
In this paper, we will employ the KdV-type bilinear equations to try  understanding the relations between
having a 3SS and being integrable.

The bilinear method drew the attention of Mikio Sato, who had seminars with  Hirota, Satsuma,
and together with Jimbo and Miwa  \cite{Kod-2023}.
This leads to Sato's observation on the connection between the bilinear Kadomtsev-Petviashvili (KP) equation
and the Pl\"ucker relations on the Grassmannian \cite{Sato-RIMS-1981}.
Soon after, Sato's KP theory was developed by the Kyoto group to become a universal theory for integrable systems.
In this theory, vertex operator plays a remarkable role, which reveals the algebraic structures of
$\tau$ functions as well as integrable equations.
The name ``vertex operator'' comes from its geometric interpretation in string theory:
it is an operator inserted at a point (a vertex) in spacetime where a particle is created or annihilated.
In 1978,  Lepowsky and Wilson \cite{LepW-1978} presented a vertex operator
that is isomorphic to the affine Lie algebra $A^{(1)}_1$.
In 1981, Date, Kashiwara and Miwa \cite{DatKM-1981}
deformed the operator and connected it with the $\tau$ function of the KdV equation,
which led to a series of beautiful work on transformation groups and integrable systems \cite{MiwJD-2000}.
In this paper, we will illustrate how vertex operators play roles in generating $\tau$ functions
for the KdV and KP hierarchies.

This paper presents a partial review of  the bilinear method,
focusing on the integrability based on the 3SS condition
and the transformation between $\tau$ functions by the BTs and  vertex operators.
We mainly referred to two monographs: one is \cite{Hir-book-2004} by Hirota,
and the other is \cite{MiwJD-2000} by Miwa, Jimbo and Date.
We will also include some recent progress of the bilinear method.

The paper is organized as follows.
In Section \ref{sec-2} we introduce basic notations, notions and properties of Hirota's bilinear derivatives.
Then we employ the KdV equation and KP equation as examples to show the procedure of  deriving
of multi-soliton solutions  from their  bilinear forms.
We also illustrate the asymptotic analysis of interactions of solitons.
In Section \ref{sec-3},
first, we show the existence of the 1SS and 2SS for a general KdV-type bilinear equation.
Then we explain the property of elastic scattering of multi-solitons and its connection with the 3SS condition.
This may help us to understand why the 3SS condition implies integrability for the KdV-type bilinear equations.
In Section \ref{sec-4} we introduce bilinear BTs and their connections with Lax pairs and superposition formulae.
In Section \ref{sec-5} we show how vertex operators  generate $\tau$ functions
for the KdV and KP hierarchies.
Finally, some remarks are given in Section \ref{sec-6}.

\section{Introduction to the bilinear method: The KdV and KP(II)}\label{sec-2}

In this section,  we introduce basic notations, notions and properties of Hirota's bilinear derivatives.
The KdV equation and KP equation will serve as examples to illustrate the procedure of deriving multi-soliton solutions
and the asymptotic analysis of interaction of solitons.

\subsection{Bilinear operator}\label{sec-2-1}

Let us start from Hirota's bilinear operator that he introduced in \cite{Hir-1974} in 1974.

\begin{definition}\label{def-1-1}
For $C^{\infty}$ differential functions $f(x,y)$ and $g(x,y)$ defined on $\mathbb{R}^2$, their \textbf{bilinear derivatives} are defined as:
\begin{equation}
D^m_xD^n_y f(x,y)\cdot g(x,y)=(\partial_x-\partial_{x'})^m (\partial_y-\partial_{y'})^n f(x,y)g(x',y')|_{x'=x,y'=y}.
\label{D}
\end{equation}
where $D$ is called \textbf{Hirota's bilinear operator}. Here  $\partial_x=\frac{\partial}{\partial x}$.
\end{definition}

The definition \eqref{D} can also be introduced as follows \cite{Hir-1977a}:
\begin{equation}
e^{\epsilon D_x+\kappa D_y} f(x,y)\cdot g(x,y)=f(x+\epsilon,y+\kappa) g(x-\epsilon,y-\kappa).
\label{D-exp}
\end{equation}
In fact, expanding both sides at $(\epsilon, \kappa)=(0,0)$ and comparing coefficients of
the power $\epsilon^m \kappa^n$, one can get definition \eqref{D}.
Here are simple examples of bilinear derivatives:
\begin{align*}
& D_x f\cdot g = f_x g-fg_x,\\
& D^2_x f\cdot g = f_{xx} g-2f_xg_x+ fg_{xx},\\
& D^3_x f\cdot g = f_{xxx} g-3f_{xx}g_x+ 3f_xg_{xx}-fg_{xxx},\\
& D_xD_y f\cdot g = f_{xy} g- f_xg_y-f_yg_x+ fg_{xy},\\
& D^m_x f\cdot g = \sum^{m}_{j=0}(-1)^j \left(\begin{smallmatrix} m\\ j \end{smallmatrix}\right)
  (\partial^{m-j}_x f(x)) (\partial^j_x g(x)),
\end{align*}
where $\left(\begin{smallmatrix} m\\ j \end{smallmatrix}\right)=\frac{m!}{(m-j)! j!}$.
From these examples one can see the difference between the bilinear derivatives of $f$ and $g$
and Leibniz's rule of the $m$-th order
derivative  $\partial^m_x(fg)$.

Bilinear derivatives  admit symmetric and bilinear properties:
\begin{align*}
& D^m_x f\cdot g = (-1)^m D^m_x g\cdot f,\\
& (aD^m_x+ b D^n_y)f\cdot g=aD^m_x f\cdot g + b D^n_y f\cdot g,\\
& D^m_x D^n_y (af+bg)\cdot h=a  D^m_x D^n_y  f\cdot h + b  D^m_x D^n_y g\cdot h,
\end{align*}
where $a,b\in \mathbb{C}$, and particularly,
\[ D^m_x f\cdot 1 = \partial^m_x f(x). \]
In addition, for linear exponential functions, e.g. $e^{kx+\omega t}$ $(k,\omega \in \mathbb{C})$,
there is a simple formula for their bilinear derivatives:
\begin{align}
 D^m_x D^n_y e^{\eta_1}\cdot e^{\eta_2} = (k_1-k_2)^m(\omega_1-\omega_2)^ne^{\eta_1+\eta_2},
\end{align}
where
\begin{equation}
\eta_i=k_i x+\omega_i y+\eta^{(0)}_i,~~ k_i,\omega_i,\eta^{(0)}_i\in \mathbb{C}.
\label{eta-i}
\end{equation}

The above definition and properties can be extended into arbitrary dimensions.
Suppose that
$ \mathbf{t}=(t_1, t_2, \cdots, t_s),~ \mathbf{p}=(p_1, p_2, \cdots, p_s), ~\mathbf{q}=(q_1, q_2, \cdots, q_s)$
are vectors in $\mathbb{R}^s$.
Then, denoting $\mathbf{p}\cdot\mathbf{t}=\sum_{i=1}^s p_i t_i$,
and  $D_{\mathbf{t}}=(D_{t_1}, D_{t_2}, \cdots,D_{t_s})$,
one can define
\begin{equation}
e^{\mathbf{p}\cdot D_{\mathbf{t}}}f(\mathbf{t})\cdot g(\mathbf{t})=f(\mathbf{t}+\mathbf{p})g(\mathbf{t}-\mathbf{p}).
\label{D-exp-2}
\end{equation}
Suppose $P(\mathbf{t})$ is a polynomial of $\mathbf{t}$, and introduce $P(D_{\mathbf{t}})$.
For example, if $P(\mathbf{t})=3t_1^2t_2+2 t_2 t_3$,
then $P(D_{\mathbf{t}})=3D_{t_1}^2D_{t_2}+2 D_{t_2}D_{t_3}$.
For a general $P(D_{\mathbf{t}})$, it is not difficult to find
\begin{align}
& P(D_{\mathbf{t}})e^{\mathbf{p}\cdot\mathbf{t}}\cdot e^{\mathbf{q}\cdot\mathbf{t}}
=P(\mathbf{p}-\mathbf{q})e^{(\mathbf{p}+\mathbf{q})\cdot \mathbf{t}},\\
& P(D_{\mathbf{t}})e^{\mathbf{p}\cdot\mathbf{t}}\cdot 1
=P(\mathbf{p})e^{\mathbf{p}\cdot \mathbf{t}}
=P(\partial_{\mathbf{t}})e^{\mathbf{p}\cdot \mathbf{t}},
\end{align}
where $\partial_{\mathbf{t}}=(\partial_{t_1}, \partial_{t_2}, \cdots, \partial_{t_s})$.
Besides, for the linear function $\eta_i$ defined in \eqref{eta-i}, one can prove that
\begin{align}
& D^r_xD^s_y(e^{\eta_1}f(x,y))\cdot (e^{\eta_2}g(x,y))\nonumber\\
=~& e^{\eta_1+\eta_2}(D_x+k_1-k_2)^r(D_y+\omega_1-\omega_2)^s f(x,y)\cdot g(x,y),
\label{id-gauge-gen}
\end{align}
and particularly, when $\eta_1=\eta_2$, one has
\begin{equation}
D^r_xD^s_y(e^{\eta_1}f)\cdot (e^{\eta_1}g)= e^{2\eta_1}D_x^rD_y^s f\cdot g,
\label{id-gauge}
\end{equation}
which is known as the gauge property of bilinear derivatives. A more general case of these properties are
\begin{align}
P(D_{\mathbf{t}}) (e^{\mathbf{p}\cdot \mathbf{t}}f(\mathbf{t}))\cdot (e^{\mathbf{q}\cdot \mathbf{t}}g(\mathbf{t}))
= e^{(\mathbf{p}+\mathbf{q})\cdot \mathbf{t}}
P(D_{\mathbf{t}}+\mathbf{p}-\mathbf{q}) f(\mathbf{t})\cdot g(\mathbf{t}),
\end{align}
and
\begin{equation}
P(D_{\mathbf{t}})(e^{\mathbf{p}\cdot \mathbf{t}}f)\cdot (e^{\mathbf{p}\cdot \mathbf{t}} g)= e^{2\mathbf{p}\cdot \mathbf{t}}P(D_{\mathbf{t}})
f\cdot g\, .
\label{gauge-pro}
\end{equation}

\subsection{N soliton solutions}\label{sec-2-2}

In this part we show how Hirota's bilinear method works in finding NSS.
The KdV equation and KP equation will serve as examples.

\subsubsection{NSS of the KdV equation}\label{sec-2-2-1}

The KdV equation reads
\begin{equation}
u_t+6uu_x+u_{xxx}=0.
\label{kdv-eq}
\end{equation}
Note that its coefficients can be arbitrary constants without changing integrability.
Usually we consider its potential form ($u=w_x$):
\begin{equation}
w_t+3(w_x)^2+w_{xxx}=0.
\label{kdv-eq-poten}
\end{equation}
Under the transformation
\begin{equation}
u=2(\ln f)_{xx}, ~~ \hbox{i.e.}~~ w=2(\ln f)_x,
\label{tran-kdv}
\end{equation}
Eq.\eqref{kdv-eq-poten} can be written as \cite{Hir-1971}
\begin{equation}
f_{xt}f-f_xf_t + f_{xxxx}f- 4f_{xxx}f_x + 3(f_{xx})^2=0,
\label{kdv-bil-1}
\end{equation}
which is
\begin{equation}
(D_xD_t +D^4_x) f\cdot f=0
\label{kdv-bil-D}
\end{equation}
in terms of the bilinear operator $D$ given in \eqref{D}.
The above equation is called the bilinear form of the KdV equation \eqref{kdv-eq}, or the bilinear KdV equation;
the solution $f$ is known as the $\tau$ function of the KdV equation; once we have $f$,
we can get back a solution of the KdV equation through the transformation \eqref{tran-kdv}.

\textit{Note: In 1971 Ryogo Hirota first introduced the bilinear method  to derive NSS of
the KdV equation \cite{Hir-1971}.
At that time the bilinear form of the KdV equation was written as \eqref{kdv-bil-1};
the bilinear operator $D$ as defined in Definition \ref{def-1-1} was  introduced later by Hirota
in\cite{Hir-1974} in 1974.}

\vskip 8pt

It is easy to check
\begin{equation}
f=1+e^{\eta},~~ \eta=kx-k^3 t +\eta^{(0)}, ~~ (k, \eta^{(0)}\in \mathbb{R})
\label{1ss-f}
\end{equation}
satisfies the bilinear KdV equation \eqref{kdv-bil-1}.
To achieve more solutions, one can (perturbatively) expand
\begin{equation}
f=1+\sum^{\infty}_{i=1} f^{(i)} \varepsilon^i,
\label{f-expand}
\end{equation}
where subscript $(i)$ is for numbering coefficients.
Substituting the above into \eqref{kdv-bil-D} and comparing coefficients of each power of $\varepsilon$,
we reach to an equation system
\begin{subequations}\label{f-eqs-kdv}
\begin{align}
& \varepsilon^1: ~~~ (\partial_{xt}+\partial^4_x) f^{(1)}=0,\label{f-eqs-kdv-a}\\
& \varepsilon^2: ~~~ (\partial_{xt}+\partial^4_x) f^{(2)}=-\frac{1}{2}(D_xD_t +D^4_x) f^{(1)}\cdot f^{(1)},\label{f-eqs-kdv-b}\\
& \varepsilon^3: ~~~ (\partial_{xt}+\partial^4_x) f^{(3)}=-(D_xD_t +D^4_x) f^{(1)}\cdot f^{(2)},\label{f-eqs-kdv-c}\\
& \varepsilon^4: ~~~ (\partial_{xt}+\partial^4_x) f^{(4)}=-(D_xD_t +D^4_x) (f^{(1)}\cdot f^{(3)}+\frac{1}{2} f^{(2)}\cdot f^{(2)}),\label{f-eqs-kdv-d}\\
& ~~~ \cdots \cdots .\nonumber
\end{align}
\end{subequations}
For \eqref{f-eqs-kdv-a} we can take $f^{(1)}=e^{\eta_i}$, where
\begin{equation}
\eta_i=k_ix-k_i^3 t +\eta_i^{(0)},~~ k_i, \eta_i^{(0)}\in \mathbb{R}.
\label{eta-i-kdv}
\end{equation}
Since \eqref{f-eqs-kdv-a} is a homogeneous linear equation, for any positive integer $N$,
\begin{equation}
f^{(1)}=\sum^N_{i=1}e^{\eta_i}
\label{f1-kdv}
\end{equation}
gives a solution to \eqref{f-eqs-kdv-a} as well,
where $\{\eta_i\}$ are defined by \eqref{eta-i-kdv}.

Now, let us look at in more details the equation system \eqref{f-eqs-kdv} with $f^{(1)}$ given as in \eqref{f1-kdv}.
When $N=1$, $f^{(1)}=e^{\eta_1}$ satisfies \eqref{f-eqs-kdv-a};
meanwhile, the rest equations in \eqref{f-eqs-kdv} hold when taking $f^{(2)}=f^{(3)}=\cdots =0$.
Thus,
\begin{equation}
f=1+\varepsilon e^{\eta_1}
\label{1ss-f-1}
\end{equation}
provides a solution to the bilinear KdV equation \eqref{kdv-bil-D}.
This means the ``perturbation'' formula \eqref{f-expand} can be truncated as shown in \eqref{1ss-f-1}.
As a consequence,  $\varepsilon$ can be arbitrary in \eqref{1ss-f-1}.
Taking $\varepsilon=1$ in \eqref{1ss-f-1} and substituting it into \eqref{tran-kdv} we have
\begin{equation}
u=2(\ln f)_{xx}=2[\ln (1+e^{\eta_1})]_{xx},
\label{kdv-1ss-1}
\end{equation}
which is a 1SS of the KdV equation.

When $N=2$, from \eqref{f1-kdv} we have $f^{(1)}=e^{\eta_1} + e^{\eta_2}$, and then
solving equation \eqref{f-eqs-kdv-b} we find
\begin{equation}
f^{(2)}=A_{12} e^{\eta_1+\eta_2},~~ A_{12}=\Bigl(\frac{k_1-k_2}{k_1+k_2}\Bigr)^2.
\label{f2-kdv}
\end{equation}
Meanwhile, $f^{(j)}=0~(j=3,4,\cdots)$ solve the rest equations in \eqref{f-eqs-kdv}.
Thus,
\begin{equation}
f=1+\varepsilon(e^{\eta_1} + e^{\eta_2})+\varepsilon^2A_{12}\, e^{\eta_1+\eta_2}
\label{2ss-f}
\end{equation}
provides a second solution to the bilinear KdV equation \eqref{kdv-bil-D}, which
will lead to a 2SS for the KdV equation via \eqref{tran-kdv}.

When $N=3$, from \eqref{f1-kdv} we can find a solution for \eqref{f-eqs-kdv}:
\begin{align}
f=& 1+\varepsilon(e^{\eta_1} + e^{\eta_2} + e^{\eta_3})\nonumber\\
  & +\varepsilon^2(A_{12} e^{\eta_1+\eta_2}+A_{13} e^{\eta_1+\eta_3}+A_{23} e^{\eta_2+\eta_3})\nonumber\\
  & +\varepsilon^3 A_{12}A_{13}A_{23} e^{\eta_1+\eta_2+\eta_3},~~~ A_{ij}=\Bigl(\frac{k_i-k_j}{k_i+k_j}\Bigr)^2.
\label{3ss-f}
\end{align}
Again, through \eqref{tran-kdv} it gives a 3SS to the KdV equation.
Note that Hirota got the above 3SS in May of 1971 \cite{Hir-2004}.

For a general number $N$, Hirota presented the following compact form:
\begin{equation}
f=\sum_{\mu=0,1} \mathrm{exp}\left(\sum^{N}_{j=1} \mu_j \eta_j+\sum^N_{1\leq i<j}\mu_i\mu_j a_{ij}\right),
\label{Nss-f}
\end{equation}
where $\eta_j$ is defined as in \eqref{eta-i-kdv}, $e^{a_{ij}}=A_{ij}$, and the summation of $\mu$ means to take
all possible $\mu_j=\{0,1\}$ $(j=1,2,\cdots, N)$.
NSS of the KdV equation is given by \eqref{tran-kdv}.
Eq.\eqref{Nss-f} is a truncated form of \eqref{f-expand} (we have taken $\varepsilon=1$ since  $\varepsilon$
is actually arbitrary).
A proof of \eqref{Nss-f} satisfying \eqref{kdv-bil-D} can be found in \cite{Hir-1971} or \cite{AblS-1981}
or \cite{Chen-book-2006}.

\subsubsection{NSS of the KP(II) equation}\label{sec-2-2-2}

The KP equation reads
\begin{equation}
(4u_t+6uu_x+u_{xxx})_x + 3\sigma u_{yy}=0,~~ (\sigma=\pm 1);
\label{kp-eq}
\end{equation}
when $\sigma=1$ it is known as the KP(II) equation and when $\sigma=-1$ it is the KP(I) equation.
We solve the KP(II) equation. By the transformation
\begin{equation}
u=2(\ln f)_{xx},
\label{tran-kp}
\end{equation}
the KP(II) equation is bilinearized as \cite{Sat-1976}
\begin{equation}
(4D_xD_t +D^4_x+3 D_y^2) f\cdot f=0.
\label{kp-bil-D}
\end{equation}
With the expension
\begin{equation}
f=1+\sum^{\infty}_{i=1} f^{(i)} \varepsilon^i,
\label{f-expand-kp}
\end{equation}
the bilinear KP(II) equation \eqref{kp-bil-D} yields
\begin{subequations}\label{f-eqs-kp}
\begin{align}
& \varepsilon^1: ~~~ (4\partial_{xt}+\partial^4_x+3\partial_y^3) f^{(1)}=0,\label{f-eqs-kp-a}\\
& \varepsilon^2: ~~~ (4\partial_{xt}+\partial^4_x+3\partial_y^3) f^{(2)}=-\frac{1}{2}(4D_xD_t +D^4_x+3 D_y^2) f^{(1)}\cdot f^{(1)},\label{f-eqs-kp-b}\\
& ~~~ \cdots \cdots .\nonumber
\end{align}
\end{subequations}
For \eqref{f-eqs-kp-a} we can take $f^{(1)}=e^{kx+hy+\omega t }$, and $k,h,\omega$ satisfy
the dispersion relation
\begin{equation}
4k\omega +k^4 + 3h^2=0.
\end{equation}
Compared with the KdV equation, here we have two free parameters.
For a better parametrisation, we introduce $h=ak$, by which we have $\omega=-k(k^2+3a^2)/4$.
Here $a$ is an arbitrary parameter, as free as $k$.
In practice we take $k=p-q,~ a=p+q$, and it follows that
\begin{equation}
k=p-q,~~ h=p^2-q^2,~~ \omega=-(p^3-q^3).
\label{kp-DR}
\end{equation}
Thus, we have
\begin{equation}
f^{(1)}=e^{\eta_1},
\label{f1-kp}
\end{equation}
where
\begin{equation}
\eta_i=(p_i-q_i)x+ (p_i^2-q_i^2)y - (p_i^3-q_i^3) t +\eta_i^{(0)}, ~~ (p_i, q_i, \eta_i^{(0)}\in \mathbb{R}).
\label{eta-i-kp}
\end{equation}
Then 1SS of the KP(II) equation can be obtained as (taking $\varepsilon=1$)
\begin{equation}
u=2(\ln f)_{xx}=\frac{(p_1-q_1)^2}{2} \mathrm{sech}^2 \frac{\eta_1}{2}.
\label{kp-1ss-1}
\end{equation}

To get more solutions, similar to the KdV equation, we take
\begin{equation}
f^{(1)}=\sum^N_{i=1}e^{\eta_i},
\label{f1-kp-2}
\end{equation}
where $\eta_i$ is defined by \eqref{eta-i-kp}.
When $N=2$, corresponding to 2SS, we find (taking $\varepsilon=1$)
\begin{subequations}\label{kp-2ss-f}
\begin{equation}
f=1+ e^{\eta_1} + e^{\eta_2}+ A_{12}\, e^{\eta_1+\eta_2},
\label{2ss-f-kp}
\end{equation}
which solves  \eqref{kp-bil-D}, where
\begin{equation}
A_{ij}=\frac{(p_i-p_j)(q_i-q_j)}{(p_i-q_j)(q_i-p_j)}.
\label{Aij-kp}
\end{equation}
\end{subequations}
Continuing this procedure, for 3SS, $f$ has the same structure as \eqref{3ss-f}. For NSS, there is
\begin{equation}
f=\sum_{\mu=0,1} \mathrm{exp}\left(\sum^{N}_{j=1} \mu_j \eta_j+\sum^N_{1\leq i<j}\mu_i\mu_j a_{ij}\right),
\label{Nss-f-kp}
\end{equation}
where $\eta_j$ is defined as in \eqref{eta-i-kp}, $e^{a_{ij}}=A_{ij}$ which is defined in \eqref{Aij-kp},
and the summation of $\mu$ means to take
all possible $\mu_j=\{0,1\}$ $(j=1,2,\cdots, N)$.
Then, NSS of the KP(II) equation is obtained via \eqref{tran-kp}.

\subsection{Asymptotic analysis of 2SS}\label{sec-2-3}

Classical solitons are characterized by their particle behavior, namely,
the elastic scattering property: each soliton keeps its own amplitude  and velocity after interaction.
In the following, for the KdV equation and the KP equation,
we illustrate such  scattering properties with mathematical explanations.

\subsubsection{The KdV equation}\label{sec-2-3-1}

First, let us look at 1SS \eqref{kdv-1ss-1} of the KdV equation, i.e.
\begin{equation}
u=\frac{k_1^2}{2} \mathrm{sech}^2 \eta_1,~~ \eta_1=k_1x-k_1^3t+\eta^{(0)}_1,
\label{kdv-1ss-2}
\end{equation}
which describes a solitary wave as depicted in Figure \ref{fig-1-1-a}.
\begin{figure}[!ht]
\centering \subfigure[]
{\label{fig-1-1-a} 
\includegraphics[width=2.5in]{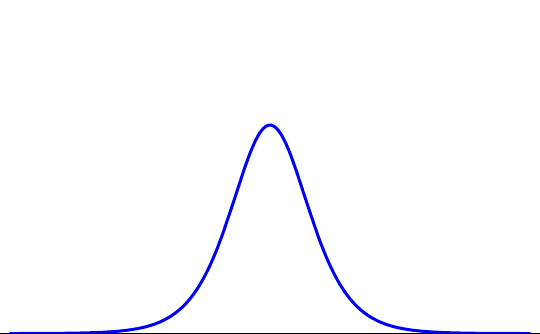}
}~~~~ \subfigure[]{
\label{fig-1-1-b} 
\includegraphics[width=2.0in]{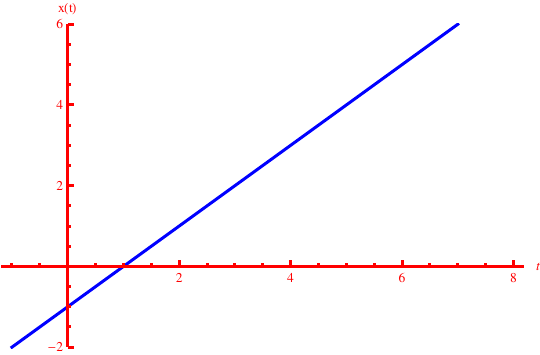}
}\\
\caption{{\small (a) 1SS of the KdV equation.~
(b) Trajectory of the vortex of 1SS: $\eta_1=0$, where $k_1=1, ~\eta_1^{(0)}=0$. \label{fig-1-1}}}
\end{figure}
The maximum of the wave, i.e. the amplitude, which is $\frac{k^2_1}{2}$, occurs when $\eta_1=0$.
When
$\eta_1=0$, i.e.
\begin{equation}
x(t)=k^2_1 t -\frac{\eta^{(0)}_1}{k_1},
\label{eta1=0}
\end{equation}
we have a straight line as depicted in Figure \ref{fig-1-1-b}, denoting  the  trajectory of the vertex of the wave.
From \eqref{eta1=0} we have
$x'(t)=k^2_1$, which stands for the velocity of the wave.
One can see that the velocity is always positive, which means the solitary wave described by the KdV equation is of single direction.

2SS of the KdV equation exhibits elastic scattering behavior, as depicted in Figure \ref{fig-1-2} and Figure \ref{fig-1-3}.
The 2SS is
\begin{subequations}\label{kdv-2ss}
\begin{align}
& u=2(\ln f)_{xx},\label{kdv-2ss-a}\\
& f=1+e^{\eta_1} + e^{\eta_2}+A_{12}\, e^{\eta_1+\eta_2},\label{kdv-2ss-b}
\end{align}
where
\begin{equation}
\eta_i=k_ix-k_i^3 t +\eta_i^{(0)},~~  A_{12}=\Bigl(\frac{k_1-k_2}{k_1+k_2}\Bigr)^2.\label{kdv-2ss-c}
\end{equation}
\end{subequations}
\begin{figure}[!ht]
\centering \subfigure[]
{\label{fig-1-2-a} 
\includegraphics[width=2.0in]{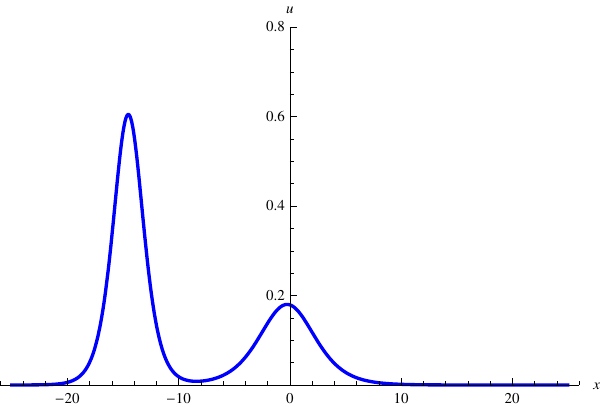}
}~~\subfigure[]{
\label{fig-1-2-b} 
\includegraphics[width=2.0in]{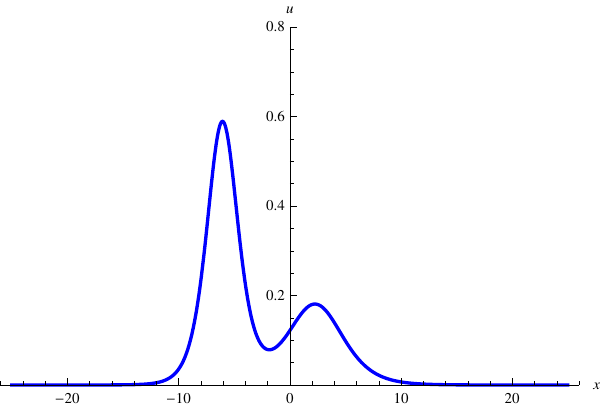}
}\\ \subfigure[]
{\label{fig-1-2-c} 
\includegraphics[width=2.0in]{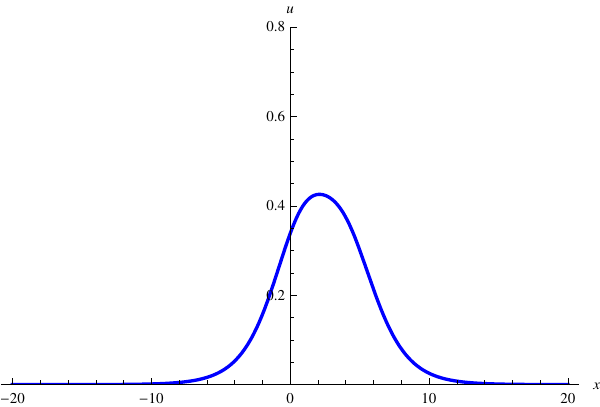}
}~~\subfigure[]
{\label{fig-1-2-d} 
\includegraphics[width=2.0in]{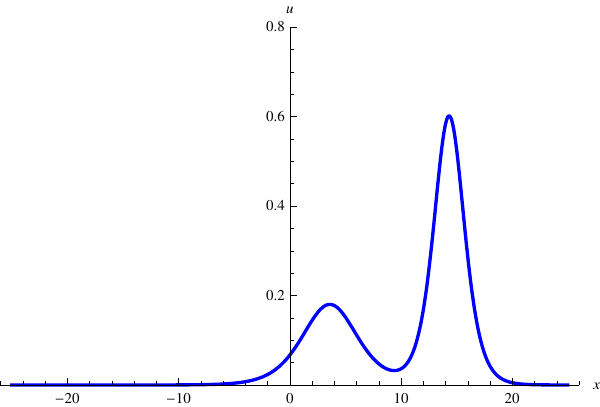}}\\
\caption{{\small 2SS \eqref{kdv-2ss} of the KdV equation ($k_1=1.1,~ k_2=0.6,~ \eta^{(0)}_1= \eta^{(0)}_2=0$).}~~
{\small
(a) $t=-12$,~ (b) $t=-5$,~ (c) $t=1$,~ (d) $t=10$.} \label{fig-1-2}}
\end{figure}
\begin{figure}[!ht]
\centering
\includegraphics[width=2.5in]{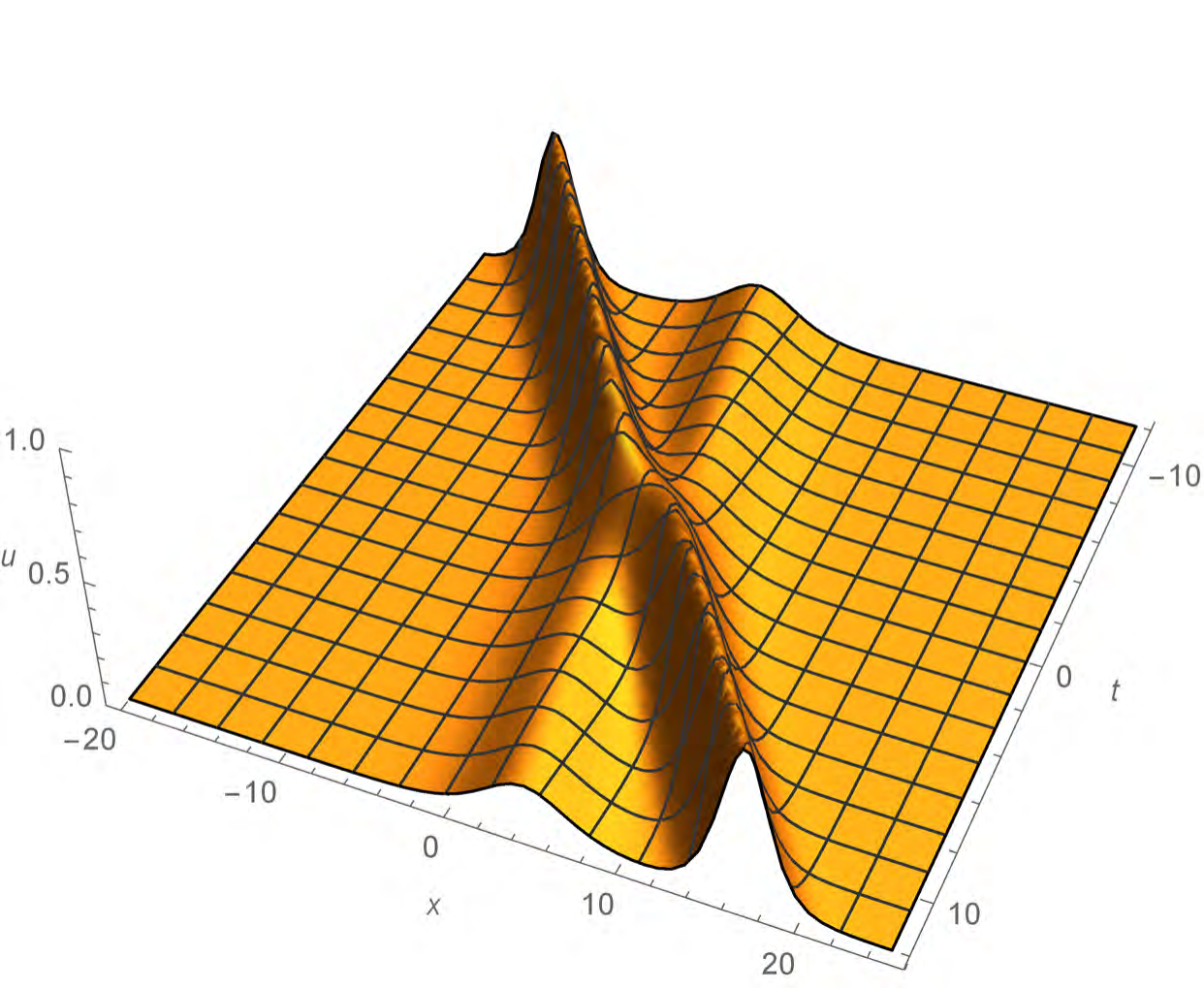}
\\
\caption{{\small  2SS \eqref{kdv-2ss} of the KdV equation ($k_1=1.1,~ k_2=0.6,~ \eta^{(0)}_1= \eta^{(0)}_2=0$).} \label{fig-1-3}}
\end{figure}

Such a scattering behavior can be understood mathematically using asymptotic analysis,
of which the generic procedure is the following.

We consider $f$ in \eqref{kdv-2ss-b} and assume $k_1>k_2>0$ without loss of generality.
Suppose $\eta_1=c$ ($c$ is certain constant) so that we can observe the 2SS along the straight line $\eta_1=c$.
To do that, we rewrite \eqref{kdv-2ss-b} in the new coordinate system $(\eta_1,t)$:
\begin{subequations}\label{f2-(eta-t)-1}
\begin{equation}
f=1+e^{\eta_1} + e^{\eta_2}+A_{12}\, e^{\eta_1+\eta_2},\label{f2-(eta-t)-1-a}
\end{equation}
where
\begin{equation}
 e^{\eta_2}=\mathrm{exp}\Bigl[{\frac{k_2}{k_1}\eta_1+k_2(k_1^2-k_2^2)t+\eta^{(0)}_2-\frac{k_2}{k_1}\eta^{(0)}_1}\Bigr].\label{f2-(eta-t)-1-b}
\end{equation}
\end{subequations}
Because of $k_1>k_2>0$, we find
\[e^{\eta_2} \sim \left\{
\begin{array}{ll}
0, & t\to -\infty,\\
+\infty, & t\to +\infty.
\end{array}\right.
\]
Therefore  in $(\eta_1,t)$ we have
\begin{equation}
f \sim \left\{
\begin{array}{ll}
1+ e^{\eta_1}, & t\to -\infty,\\
e^{\eta_2}(1+A_{12}e^{\eta_1}), & t\to +\infty.
\end{array}\right.
\label{f-asy-2ss}
\end{equation}
Due to the gauge property of bilinear derivatives,
the factor $e^{\eta_2}$ in the above does not change solutions of the bilinear KdV equation
\eqref{kdv-bil-D}, and does not change the value of $u=2(\ln f)_{xx}$ either.
Thus, if we observe 2SS along the straight line $\eta_1=c$, when $t\to -\infty$ we only see the 1SS
\begin{subequations}\label{kdv-sol-k1}
\begin{equation}
u=2[\ln(1+e^{\eta_1})]_{xx};
\label{kdv-sol-k1-a}
\end{equation}
and when $t \to +\infty$, we see
\begin{equation}
u=2[\ln(1+A_{12}e^{\eta_1})]_{xx}.
\label{kdv-sol-k1-b}
\end{equation}
\end{subequations}
This is still the original 1SS \eqref{kdv-sol-k1-a} (after interaction it has the amplitude and velocity as same as before)
but there is a phase shift $-\frac{2}{k_1}\ln \frac{k_1-k_2}{k_1+k_2}$.

Similarly, in the coordinate system $(\eta_2,t)$ we can see
\[e^{\eta_1} \sim \left\{
\begin{array}{ll}
+\infty, & t\to -\infty,\\
0, & t\to +\infty.
\end{array}\right.
\]
Then we have
\begin{equation}
f \sim \left\{
\begin{array}{ll}
e^{\eta_1}(1+A_{12}e^{\eta_2}), & t\to -\infty,\\
1+ e^{\eta_2}, & t\to +\infty.
\end{array}\right.
\end{equation}
Thus we can see that the soliton determined by $k_2$ keeps its  amplitude and velocity before and  after the interaction
but gains a phase shift  $\frac{2}{k_2}\ln  \frac{k_1-k_2}{k_1+k_2}$.
Such phase shifts due to the interaction can be seen in  Figure \ref{fig-1-3}.

\subsubsection{The KP(II) equation}\label{sec-2-3-2}

The 1SS \eqref{kp-1ss-1} of the KP(II) equation is depicted in Figure \ref{fig-1-4-a}.
At any given time $t$ it exhibits like a straight line on the $(x,y)$ plane, with amplitude $(p_1-q_1)^2/2$.
The  straight line is given by $\eta_1=0$ and it also provides the velocity by which  the line moves on the $(x,y)$ plane:
\[(x'(t), y'(t))=-(p^2+pq+q^2)\Bigl( 1, ~\frac{1}{p+q}\Bigr ).\]

\begin{figure}[!ht]
\centering \subfigure[]
{\label{fig-1-4-a} 
\includegraphics[width=2.5in]{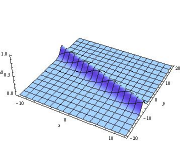}
}~~~~ \subfigure[]{
\label{fig-1-4-b} 
\includegraphics[width=2.1in]{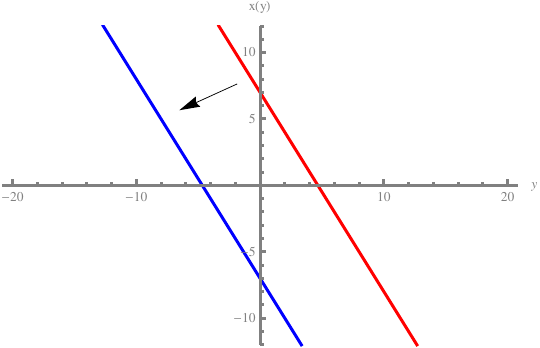}
}\\
\caption{{\small (a) 1SS of the KP(II). It is given by \eqref{kp-1ss-1}
with $(p_1, q_1, \eta^{(0)}_1)=(0.5, 1, 0), t=0$.
~
(b) Trajectory of the line soliton in (a): red line is for $t=-4$ and blue for $t=4$. \label{fig-1-4}}}
\end{figure}

For the 2SS given by \eqref{tran-kp} with \eqref{kp-2ss-f}, at a given time $t$ it behaves like two lines crossed in Figure \ref{fig-1-5}.
When $t$ is fixed we can consider $t$ as a constant and analyze the asymptotic behaviors when  $y\to \pm \infty$. When $A_{12}\neq 0$,
the procedure is similar to the KdV case in Sec.\ref{sec-2-3-1}, and here we skip it.
\begin{figure}[!ht]
\centering \subfigure[]
{\label{fig-1-5-a} 
\includegraphics[width=2.5in]{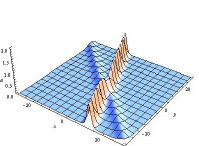}
}~~~~ \subfigure[]{
\label{fig-1-5-b} 
\includegraphics[width=2.5in]{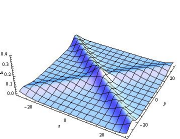}
}\\
\caption{{\small 2SS of the KP(II) equation. It is given by \eqref{tran-kp} with \eqref{kp-2ss-f}: (a) $(p_1, q_1, \eta^{(0)}_1)=(0.8, 0.2, 0)$,
$(p_2, q_2, \eta^{(0)}_2)=(-0.5, 0.9, 0)$, $t=0$;
~
(b) $(p_1, q_1, \eta^{(0)}_1)=(-0.4, 0.7, 0)$,
$(p_2, q_2, \eta^{(0)}_2)=(0.8, 0.4, 0)$, $t=0$. \label{fig-1-5}}}
\end{figure}

Let us consider the case of $A_{12}=0$.
Recalling in Sec.\ref{sec-2-3-1} for the KdV equation, its 1SS \eqref{kdv-1ss-2} is completely determined by $k_1$;
in 2SS \eqref{kdv-2ss} if $k_1=k_2$ then $A_{12}=0$ and the 2SS degenerates to 1SS.
However, for the KP(II) equation, its 1SS is determined by two parameters, $p_1$ and $q_1$.
Particularly, based on the special form of $A_{12}$, i.e. \eqref{Aij-kp}:
\begin{equation}
A_{12}=\frac{(p_1-p_2)(q_1-q_2)}{(p_1-q_2)(q_1-p_2)},
\label{A12-kp}
\end{equation}
when $p_1\neq p_2$ but $q_1=q_2$, we have $A_{12}=0$ and  consequently \eqref{2ss-f-kp} degenerates to
\begin{equation}
f=1+ e^{\eta_1} + e^{\eta_2},
\label{2ss-f-kp-res}
\end{equation}
where $\eta_i$ is given as \eqref{eta-i-kp}.
In this case, the 2SS does not degenerate to 1SS but exhibits resonance of two line solitons, as described in Figure \ref{fig-1-6}.
\begin{figure}[!ht]
\centering
\includegraphics[width=2.5in]{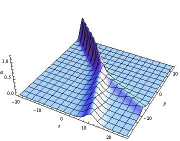}
\\
\caption{{\small  2SS resonance of the KP(II) equation. It is given by \eqref{tran-kp} with \eqref{kp-2ss-f}, where $p_1=1.0,  p_2=-0.2$,
$q_1=q_2=0.5$, $\eta^{(0)}_1= \eta^{(0)}_2=0$, $t=0$.} \label{fig-1-6}}
\end{figure}

Resonance of solitons can be understood as a behavior occurring when some parameters tend to be same ($q_1=q_2$ can be considered as a result of $q_2\to q_1$).
Such a phenomena  of solitary waves has been studied by   Zakharov and Shabat \cite{ZakS-1972},
and Miles \cite{Mil-1978}, etc.

Now let us give an asymptotic analysis for the resonance described in Figure \ref{fig-1-6}.
Consider the simplified case in which we take $t=0$ and $\eta^{(0)}_i=0$.
Thus in \eqref{2ss-f-kp-res}, there is
\[\eta_i=k_i x+h_i y,~~ k_i=p_i-q_i,~ h_i=p_i^2-q_i^2.\]
We rewrite $\eta_2$ in terms of the coordinate  $(\eta_1, y)$:
\[\eta_2=\frac{k_2}{k_1}\eta_1+\frac{1}{k_1}(k_1h_2-k_2h_1)y.\]
Using the data in Figure \ref{fig-1-6} (noticing that $k_1>0, ~k_1h_2-k_2h_1>0$),
we find
\[e^{\eta_2} \sim \left\{
\begin{array}{ll}
0, & y\to -\infty,\\
+\infty, & y\to +\infty.
\end{array}\right.
\]
Thus, in the coordinate  $(\eta_1,y)$ we have
\[
f \sim \left\{
\begin{array}{ll}
1+ e^{\eta_1}, & y\to -\infty,\\
e^{\eta_2}, & y\to +\infty,
\end{array}\right.
\]
and for $u$ in $(\eta_1,y)$ we can see the following,
\[
u \sim \left\{
\begin{array}{ll}
\frac{k_1^2}{2}\mathrm{sech}^2 \eta_1, & y\to -\infty,\\
0, & y\to +\infty.
\end{array}\right.
\]
In a similar way, if in the coordinate   $(\eta_2,t)$ we can see that
\[
u \sim \left\{
\begin{array}{ll}
\frac{k_2^2}{2}\mathrm{sech}^2 \eta_2, & y\to -\infty,\\
0, & y\to +\infty.
\end{array}\right.
\]
These can explain the two ``legs'' in Figure \ref{fig-1-6} and there are no solitons traveling along
the previous directions when $y\to +\infty$.

Next, let us figure out the characteristic of the soliton when $y\to +\infty$.
In the  coordinate system  $(\eta_1-\eta_2, y)$ we  find (noticing that with the date in Figure \ref{fig-1-6} we have $k_1-k_2>0, ~k_1h_2-k_2h_1>0$)
\[e^{\eta_i}=e^{\frac{k_i}{k_1-k_2}(\eta_1-\eta_2)+\frac{k_1h_2-k_2h_1}{k_1-k_2}y}
\sim \left\{
\begin{array}{ll}
0, & y\to -\infty,\\
+\infty, & y\to +\infty.
\end{array}\right.
\]
Then, in the coordinate system $(\eta_1-\eta_2,y)$ there is
\[
f=1+e^{\eta_2}(1+e^{\eta_1-\eta_2}) \sim \left\{
\begin{array}{ll}
0, & y\to -\infty,\\
e^{\eta_2}(1+e^{\eta_1-\eta_2}), & y\to +\infty.
\end{array}\right.
\]
Thus, for $u$ in $(\eta_1-\eta_2,y)$ we can see that
\[
u \sim \left\{
\begin{array}{ll}
0, & y\to -\infty,\\
\frac{(k_1-k_2)^2}{2}\mathrm{sech}^2 (\eta_1-\eta_2), & y\to +\infty.
\end{array}\right.
\]

\textit{Note: Asymptotic analysis is helpful to understanding interactions of solitons and explaining special scattering behaviors.
For more examples of multi-soliton interactions and their asymptotic analysis, one can refer to the review paper \cite{Hie-2002}
by Jarmo Hietarinta.
For the various resonances of line solitons of the KP(II) equation and their related interesting
mathematical structures, one can refer to \cite{Kod-2010,KodW-2011a,KodW-2011b}, etc,  by Yuji Kodama
and his collaborator.}

\section{Hirota's integrability, elastic scattering and 3SS condition}\label{sec-3}

In this section we will focus on the KdV-type bilinear equations
which always admit 1SS and 2SS, no matter whether they are integrable or not.
We will introduce integrability of this type bilinear equations in Hirota's sense.
Then we will explain the property of elastic scattering of multi-solitons and its connection with the 3SS condition.
This may help us to understand why the 3SS condition implies integrability for the KdV-type bilinear equations.

\subsection{2SS of the KdV-type bilinear equations}\label{sec-3-1}

It is amazing that the bilinear equations of the KdV type automatically admit  1SS and 2SS.

\subsubsection{Bilinear equations of the KdV type}\label{sec-3-1-1}

Consider the following bilinear equation
\begin{equation}
P(D_{\mathbf{t}})f\cdot f=0,
\label{kdvtype-bil}
\end{equation}
where $P$ is an even polynomial, i.e. $P(\mathbf{t})=P(-\mathbf{t})$ and satisfying $P(\mathbf{0})=0$.
Such a bilinear equation  \eqref{kdvtype-bil} is called  a bilinear equation of the KdV type \cite{Hir-1980,Hie-1987a}.
Assume that
\begin{subequations}\label{kdv-type-1ss}
\begin{equation}
f=1+e^{\eta_1},
\label{kdv-type-1ss-a}
\end{equation}
where
\begin{equation}
\eta_1=\mathbf{p_1}\cdot \mathbf{t}+\eta^{(0)}_1.
\label{kdv-type-1ss-b}
\end{equation}
\end{subequations}
It follows that
\[f\cdot f=1 \cdot 1 +1\cdot e^{\eta_1}+e^{\eta_1}\cdot 1 + e^{\eta_1}\cdot e^{\eta_1}.\]
Since the terms $1 \cdot 1$ and  $ e^{\eta_1}\cdot e^{\eta_1}$ varnish under the action of $P(D_{\mathbf{t}})$,
we have
\[P(D_{\mathbf{t}})f\cdot f =2P(D_{\mathbf{t}})e^{\eta_1}\cdot 1
=2 P(\partial_{\mathbf{t}})e^{\eta_1} = 2 P(\mathbf{p_1}) e^{\eta_1}.
\]
Thus, once $P(\mathbf{p_1})=0$, \eqref{kdv-type-1ss} gives a solution to \eqref{kdvtype-bil}.
We call $P(\mathbf{p_1})=0$ to be the dispersion relation  of the bilinear equation \eqref{kdvtype-bil}.

Consider
\begin{subequations}\label{kdv-type-2ss}
\begin{equation}
f=1+e^{\eta_1}+e^{\eta_2}+A_{12}e^{\eta_1+\eta_2},
\label{kdv-type-2ss-a}
\end{equation}
where $A_{12}$ is a constant to be determined,
\begin{equation}
\eta_i=\mathbf{p_i}\cdot \mathbf{t}+\eta^{(0)}_i,
\label{kdv-type-2ss-b}
\end{equation}
satisfying the dispersion relation
\begin{equation}
P(\mathbf{p_i})=0.
\label{kdv-type-2ss-c}
\end{equation}
Substituting \eqref{kdv-type-2ss-a} into   equation \eqref{kdvtype-bil} and making use of the
dispersion relation \eqref{kdv-type-2ss-c},
it is easy to see that \eqref{kdv-type-2ss-a} satisfies  equation \eqref{kdvtype-bil} provided
\begin{equation}
A_{12}=-\frac{P(\mathbf{p_1}-\mathbf{p_2})}{P(\mathbf{p_1}+\mathbf{p_2})}.
\label{kdv-type-2ss-d}
\end{equation}
\end{subequations}
It is Hirota who first found this fact \cite{Hir-1980}.
Here we note that ``automatically'' existing  2SS means there is no extra condition  on $\mathbf{p_i}$ beyond
the dispersion relation \eqref{kdv-type-2ss-c}.

\subsubsection{Other cases}\label{sec-3-1-2}

As an example let us look at the following bilinear system (\cite{Hie-1987c}, Sec.6):
\begin{subequations}\label{sgtype-bil}
\begin{align}
& B(D_{\mathbf{t}})G\cdot F=0,\label{sgtype-bil-a}\\
& A(D_{\mathbf{t}})(F\cdot F+\epsilon G\cdot G)=0,\label{sgtype-bil-b}
\end{align}
\end{subequations}
where $A$ is an even polynomial, $B$ is a polynomial
and  $\epsilon=\pm 1$. The above equation system admits a 1SS;
\begin{equation*}
F=1,~~ G=e^{\eta_1},~~ \eta_1=\mathbf{p_1}\cdot \mathbf{t}+\eta^{(0)}_1,
\end{equation*}
where $\eta_1$ satisfies dispersion relation: $B(\mathbf{p_1})=0$.
One type of 2SS of \eqref{sgtype-bil} is:
\begin{subequations}
\begin{align}
F=1-A_{12}e^{\eta_1+\eta_2},~~G=e^{\eta_1} + e^{\eta_2},
\end{align}
where $\eta_i=\mathbf{p_i}\cdot \mathbf{t}+\eta^{(0)}_i$ satisfies  the dispersion relation $B(\mathbf{p_i})=0$, and
\begin{equation}
A_{12}=-\epsilon \, \frac{A(\mathbf{p_1}-\mathbf{p_2})}{A(\mathbf{p_1}+\mathbf{p_2})}.
\end{equation}
\label{sg-type-2ss}
\end{subequations}

Not any bilinear equation (system) will automatically admit a 2SS.
For example, the following bilinear   system (the NLS-type) \cite{Hie-1988}:
\begin{subequations}\label{nlstype-bil}
\begin{align}
& B(D_{\mathbf{t}})G\cdot F=0,\label{nlstype-bil-a}\\
& A(D_{\mathbf{t}})F\cdot F=2\epsilon |G|^2,\label{nlstype-bil-b}
\end{align}
\end{subequations}
where $A$ is an even polynomial, $B$ is a polynomial, $F\in \mathbb{R}(\mathbf{t})$,
$G\in \mathbb{C}(\mathbf{t})$, $|G|^2=GG^*$ and $*$ stands for complex conjugate. It has a 1SS:
\begin{equation*}
F=1+a \, e^{\eta_1+\eta_1^*},~~ G=e^{\eta_1},
\end{equation*}
where $\eta_1=\mathbf{p_1}\cdot \mathbf{t}+\eta^{(0)}_1$, $\mathbf{p}\in \mathbb{C}^s$, $\eta_1^{(0)}\in \mathbb{C}$, $B(\mathbf{p_1})=0$ and
$a=\frac{-\epsilon}{A(\mathbf{p_1}+\mathbf{p_1}^*)}$.
However, its 2SS does not exist automatically  \cite{Hie-1988}.

\subsection{Hirota's integrability and 3SS condition}\label{sec-3-2}

\subsubsection{Hirota's integrability of the KdV type}\label{sec-3-2-1}

Take the KdV-type bilinear equation \eqref{kdvtype-bil} as an example. For such an equation, i.e.
\begin{equation}
P(D_{\mathbf{t}})f\cdot f=0,
\label{kdvtype-bil-2}
\end{equation}
it is said to be Hirota-integrable if for all positive integers $N$, it has a NSS of the form
\begin{equation}
f=1+\varepsilon \sum^{N}_{i=1}e^{\eta_i}+\{\hbox{finite number of higher-order terms in $\varepsilon$}\},
\label{f-Nss}
\end{equation}
without any further conditions on the parameters $\mathbf{p_i}$ beyond the dispersion relation
\begin{equation}
P(\mathbf{p_i})=0.
\label{DR-kdv}
\end{equation}

In general, for a bilinear equation (system), when it has a 1SS which is described by  $e^{\eta_i}$ $(\eta_i=\mathbf{p_i}\cdot \mathbf{t}+\eta^{(0)}_i)$,
there is a condition on $e^{\eta_i}$, for example, the dispersion relation, which we call 1SS-condition.
If the bilinear equation (system) allows a solution which is described by a polynomial  of
arbitrarily many $\{e^{\eta_i}\}$ (where each $e^{\eta_i}$ describes one single soliton),
and there is no extra condition on each $e^{\eta_i}$ besides the 1SS-condition,
we say the bilinear equation (system) is \textbf{Hirota integrable}.

Hirota presented the following form for the NSS of the KdV-type bilinear equation \eqref{kdvtype-bil-2} \cite{Hir-1980}:
\begin{subequations}\label{kdv-type-Nss}
\begin{equation}
f=\sum_{\mu=0,1} \mathrm{exp}\left(\sum^{N}_{j=1} \mu_j \eta_j+\sum^N_{1\leq i<j}\mu_i\mu_j a_{ij}\right),
\label{kdv-type-Nss-a}
\end{equation}
where $\eta_j=\mathbf{p_j}\cdot \mathbf{t}+\eta^{(0)}_j$ are defined as \eqref{kdv-type-2ss-b},
\begin{equation}
P(\mathbf{p_i})=0,~~~e^{a_{ij}}=A_{ij}=-\frac{P(\mathbf{p_i}-\mathbf{p_j})}{P(\mathbf{p_i}+\mathbf{p_j})},
\label{kdv-type-Nss-b}
\end{equation}
\end{subequations}
the summation of $\mu$ means to take
all possible $\mu_j=\{0,1\}$ $(j=1,2,\cdots, N)$;
and further than that,
the following condition is needed:
\begin{equation}
\sum_{\sigma=\pm 1} \left[P(\sum^{N}_{j=1} \sigma_j \mathbf{p_j})\times
\Bigl(\, \prod_{1\leq i<j\leq N}\sigma_i\sigma_jP(\sigma_i \mathbf{p_i}-\sigma_j \mathbf{p_j})\Bigr)\right]=0,
\label{kdv-type-Nss-cond}
\end{equation}
the summation of $\sigma$ means to take
all possible $\sigma_j=\{1,-1\}$ $(j=1,2,\cdots, N)$.
This condition holds automatically for the  $N=2$ case.

\subsubsection{Elastic scattering property and 3SS condition}\label{sec-3-2-2}

Those bilinear equations (systems) that automatically have 2SS may not have a 3SS;
even when they have a 3SS, they might not be Hirota integrable (there may be extra condition on $\mathbf{p_i}$).
One famous example is the
(2+1)-dimensional sG equation (see \cite{Hir-1973})
\begin{equation}
\varphi_{xx}+\varphi_{yy}-\varphi_{tt}=\sin \varphi,
\label{sg-2+1}
\end{equation}
which, by the transformation
\begin{equation}
\varphi=4 \arctan \frac{g}{f},
\label{tran-sg}
\end{equation}
is transformed into bilinear form
\begin{subequations}
\begin{align}
& (D_x^2+D_y^2-D_t^2)g\cdot f=gf, \\
& (D_x^2+D_y^2-D_t^2)(f\cdot f-g\cdot g)=0.
\end{align}
\label{sg-2+1-bil}
\end{subequations}
The above system is a special case of \eqref{sgtype-bil}
and it  has a 2SS automatically (cf. \eqref{sg-type-2ss}). Its 3SS reads \cite{Hir-1973}
\begin{subequations}\label{sg-2+1-2ss}
\begin{align}
& f= 1+ A_{12}e^{\eta_1+\eta_2}+ A_{13}e^{\eta_1+\eta_3} + A_{23}e^{\eta_2+\eta_3}, \\
& g= e^{\eta_1}+ e^{\eta_2} + e^{\eta_3} + A_{12}A_{13}A_{23} e^{\eta_1+\eta_2+\eta_3},
\end{align}
where
\begin{align}
& \eta_i=a_i x+b_i y - c_i t +\eta^{(0)}_i,\\
& \hbox{dispersion~relation:} ~ a^2_i+b_i^2-c_i^2=1,\label{sg-2+1-2ss-DR}\\
& A_{ij}=\frac{(a_i-a_j)^2+(b_i-b_j)^2-(c_i-c_j)^2}{(a_i+a_j)^2+(b_i+b_j)^2-(c_i+c_j)^2},
\end{align}
and an extra condition is needed:
\begin{equation}
\left|\begin{array}{ccc}
a_1 & a_2 & a_3\\
b_1 & b_2 & b_3\\
c_1 & c_2 & c_3
\end{array}
\right|=0.
\label{sg-2+1-2ss-cond}
\end{equation}
\end{subequations}
In the above results, for the 2SS which exists automatically, the only condition on  $\{a_i,b_i,c_i\}$ is the
dispersion relation \eqref{sg-2+1-2ss-DR};
however, if 3SS exists, in addition to the dispersion relation, the extra condition \eqref{sg-2+1-2ss-cond} is required.
Thus,   \eqref{sg-2+1-bil} is an bilinear system that possesses a 3SS
but is not integrable in Hirota's sense.
Hirota once proposed such a question in \cite{Hir-1980} for the KdV-type bilinear equations:
``Under what conditions does $P$ satisfy the identity \eqref{kdv-type-Nss-cond}?"

For the KdV-type bilinear equation \eqref{kdvtype-bil},
\textbf{3SS-condition} is specially referred to the following:
the bilinear equation \eqref{kdvtype-bil} has a 3SS,
and the condition on each $\eta_i$
is nothing beyond the 1SS-condition.
In general, it is conjectured that for  KdV-type bilinear equation \eqref{kdvtype-bil}
the 3SS-condition is equivalent to Hirota's integrability.

In the following, let us take the KdV-type bilinear equation  \eqref{kdvtype-bil}  as an example
to see how the form of 3SS is determined by the elastic scattering behavior of multi-solitons.

The property of elastic scattering of multi-solitons requires the following:
``removing a soliton from NSS, the left (N$-$1) solitons keep the elastic scattering structure of (N$-$1)SS".
Removing a soliton means the soliton is far from others\, ------ which, mathematically, can be done through two ways:
either $e^{\eta_k} \to 0 $ or $e^{\eta_k} \to \infty$ (refer to the asymptotic analysis of 2SS in Sec.\ref{sec-2-3}).
For the KdV-type bilinear equation \eqref{kdvtype-bil}, it automatically has a 2SS (see \eqref{kdv-type-2ss}):
\begin{subequations}\label{kdv-type-2ss-2}
\begin{equation}
f=1+e^{\eta_1}+e^{\eta_2}+A_{12}e^{\eta_1+\eta_2},
\label{kdv-type-2ss-2a}
\end{equation}
where
\begin{equation}
P(\mathbf{p_i})=0,
\label{kdv-type-2ss-2b}
\end{equation}
\begin{equation}
A_{ij}=-\frac{P(\mathbf{p_i}-\mathbf{p_j})}{P(\mathbf{p_i}+\mathbf{p_j})}.
\label{kdv-type-2ss-2c}
\end{equation}
\end{subequations}
With the requirement of the elastic scattering property, if there is no further condition on  $\mathbf{p_i}$ beyond \eqref{kdv-type-2ss-2b},
after analysis
(considering a general form for 3SS and assuming the 2SS \eqref{kdv-type-2ss-2a}
is a result of a 3SS after removing one soliton under the elastic scattering property)
we can find the 3SS  of equation \eqref{kdv-type-2ss-2}
can only be  the following form
\begin{align}
f=& 1+e^{\eta_1} + e^{\eta_2} + e^{\eta_3}\nonumber\\
  & + A_{12} e^{\eta_1+\eta_2}+A_{13} e^{\eta_1+\eta_3}+A_{23} e^{\eta_3+\eta_3}\nonumber\\
  & + A_{12}A_{13}A_{23} e^{\eta_1+\eta_2+\eta_3},
\label{kdvtype-3ss-f}
\end{align}
and $A_{ij}$ must be defined as \eqref{kdv-type-2ss-2c}.
If we start from the 3SS \eqref{kdvtype-3ss-f} and
assuming the above 3SS is the result of a 4SS after removing one soliton under the elastic scattering property,
we can reach to a form for 4SS.
Continuing such a procedure one can obtain  5SS, 6SS, $\cdots$.

Thus, for the KdV-type bilinear equation \eqref{kdvtype-bil},
if we only require the dispersion relation \eqref{kdv-type-2ss-2b} and the elastic scattering property,
its NSS (if it has) can only be the form  \eqref{kdv-type-Nss},
(cf. \eqref{Nss-f} for the KdV equation).

Now, if \eqref{kdv-type-Nss} provides a solution to  \eqref{kdvtype-bil}, then \eqref{kdvtype-bil} is Hirota's integrable.
However, not all the bilinear equations have their 3SSs which are only built on
the dispersion relation: $P(\mathbf{p_i})=0$.
In 1987, Jarmo Hietarinta found that the KdV-type bilinear equations that satisfy the 3SS-condition are:
\begin{subequations}\label{Kdv-eq-1987}
\begin{align}
& (D_x^4-4D_xD_t+3D_y^2)f\cdot f=0,\\
& (D_x^3D_t+a D_x^2 + D_tD_y)f\cdot f=0,\\
& [D_xD_t(D_x^2+\sqrt{3} D_xD_t+D_t^2)+ aD_x^2+ bD_xD_t+ cD_t^2] f\cdot f=0, \label{Kdv-eq-1987-c}\\
& (D_x^6+5D^3_xD_t-5D_t^2+D_xD_y) f\cdot f=0,
\end{align}
\end{subequations}
etc., where, $a,b,c$ are arbitrary constants. For more results for other types of bilinear equations,
one can refer to  \cite{Hie-1987a,Hie-1987b,Hie-1987c,Hie-1988} and \cite{Hie-2009}.
For the discrete bilinear equations of the KdV-type  on $3\times 3$ stencil,
the equations allowing 3SSs are obtained in \cite{HieZ-2013} by searching as
what have done in \cite{Hie-1987a} for the continuous case.

\subsubsection{Notes}\label{sec-3-2-3}

In \cite{Hie-1990} Hietarinta proposed a question
on the algorithmic way to nonlinearizing bilinear equations.
Recently, such a way has been developed in \cite{ZLZ-CTP-2025,LZZZ-2026}.
The algorithm is designed based on the relations between Hirota's bilinear derivatives and
the Bell polynomials, which are revealed by Gilson, Lambert and their collaborators
\cite{GLNW-PRSL-1996,LLSW-JPA-1994}.
Here we present nonlinear forms for some bilinear equations.
For the details and more examples about the nonlinearization procedure using the Bell polynomials,
one can refer to \cite{ZLZ-CTP-2025,LZZZ-2026}.
For the bilinear equation \eqref{Kdv-eq-1987-c} which is also called Hietartina's equation
(cf. \cite{ZHS-JCP-2018}), through the transformation $v=2(\ln f)_{x}$,
it corresponds to the nonlinear form \cite{ZLZ-CTP-2025}:
\begin{equation}\label{Hie-eq-n-uv}
v_{xxx}+v_{xtt}+3v_{x}(u_{x}+ v_{t})
+\sqrt{3}(v_{xxt}+u_{x}v_{t}+2v_{x}^{2})+au_{x}+bv_{x}+cv_{t}=0, ~~ v_x=u_t.
\end{equation}
This equation admits a 3SS, but its complete integrability (e.g. Lax pair, B\"acklund transformation,
Painlev\'e test) remains unclear.
A second example is a mKdV-type bilinear equation (see Table I in \cite{Hie-1987b})
\begin{subequations}\label{5.17}
\begin{align}
		&(D_{x}^{2}D_{t}+D_{y})f \cdot g=0, \label{5.17a}\\
		&D_{x}^{2} f \cdot g=0.   \label{5.17b}
\end{align}
\end{subequations}
With the transformation $w=\partial_x\ln(f/g)$, its nonlinear form reads \cite{ZLZ-CTP-2025}
\begin{equation}\label{mkdv-b-3}
		w_{y}+w_{xxt}-4w^{2}w_{t}-2w_{x}\int_{-\infty}^{x}(w^2)_{t}\,dx=0.
\end{equation}
This is an integrable equation known as the breaking-soliton mKdV equation found by Bogoyavlenskii
in 1990 (see equation (1.4) on page 47 of \cite{Bogo-1990}).
Note that Billig first realized the vertex operator corresponding to a toroidal Lie algebra
can generate tau functions and bilinear forms for the breaking-soliton KdV hierarchy \cite{Bil-1999}.
Then, the research in the language of free fermions was extended to the breaking-soliton
KdV \cite{IT-IMRN-2001}, NLS \cite{Kakei-2002} and mKdV \cite{YC-LMP-2024} equations.
The coupled system \eqref{5.17} was rederived in \cite{YC-LMP-2024}.
One more example is a sG type bilinear equation
(see \cite{Hie-1987c}, Table II):
\begin{subequations}\label{sg-3-bi}
\begin{align}
&(D_x D_y +1 ) f \cdot g =0,\label{sg-3-b}\\
&D_x D_t  (f \cdot f - g \cdot g) =0.\label{sg-3-a}
\end{align}
\end{subequations}
Introducing $w$ by $\tan \frac{w}{4}=f/g$, the above system gives rise to \cite{LZZZ-2026}
\begin{equation}\label{w_xyt}
w_{xyt} = w_x w_{xy}\cot w -  w_y w_{xt}\tan w,
\end{equation}
or
\begin{align}
\left(\frac{\theta_{\xi t}}{\sin \theta}\right)_{\xi} -\left(\frac{\theta_{\mu t}}{\sin \theta}\right)_{\mu}
+\frac{\theta_{\mu t}\theta_{\xi}-\theta_{\xi t}\theta_{\mu}}{\sin^2\theta}=0,
\end{align}
where
\begin{equation}
w(x,y,t)=2\theta(\xi=x+y, \mu=x-y,t).
\end{equation}
The later was derived by Konopelchenko and Rogers in 1991 \cite{KR-PLA-1991}
as an integrable (2+1) dimensional sG equation.
Both \eqref{5.17} and \eqref{sg-3-bi} are the bilinear equations found by Hietarinta in 1987 \cite{Hie-1987b,Hie-1987c}
by searching for admitting 3SS, while their integrability was revealed later.

In \cite{Hie-1987a,Hie-1987b,Hie-1987c,Hie-1988},
Hietarinta investigated several special types of bilinear equations
that allow having 3SS (or 2SS for the NLS type).
For the KdV-type bilinear equations,
the 3SS condition is a necessary condition for Hirota's integrability
in the sense of having NSS of the form \eqref{kdv-type-Nss}.
However, it is not generally proven to be sufficient for all types of bilinear systems,
even though it is conjectured for the bilinear equations of the KdV-type \eqref{kdvtype-bil}.
Ralph Willox described  a coupled KdV-type bilinear system (with a single tau function)
that admits a 3SS but no 4SS.
Strictly speaking, the 3SS condition (2SS for some bilinear equations, e.g.  the NLS-type)
is necessary but not universally sufficient for integrability, and  additional tests
(e.g. the Lax pair or Painlev\'e property) are often required.

\section{B\"acklund transformations}\label{sec-4}

B\"acklund transformation (BT) requires compatibility, which connects it with integrability.
Hirota developed a technique to construct bilinear BTs \cite{Hir-1974}.
In this section, we will introduce his technique and
explain the connections between  BTs and  Lax pairs.
We will also illustrate superposition formulae of BTs,
which led Hirota to finding 3SS for the sG equation (cf.\cite{Hir-2004}).

\subsection{Bilinear identities}\label{sec-4-1}

To construct bilinear BTs, first, we introduce a way to generate a large class of bilinear identities.

\begin{property}\label{pro-2-1}
The following equality holds (see \cite{Hir-1977a}, Appendix I):
\begin{align}
& e^{D_1}(e^{D_2} a\cdot b)\cdot (e^{D_3} c\cdot d)\nonumber\\
=~& e^{\frac{1}{2}(D_2-D_3)}(e^{\frac{1}{2}(D_2+D_3)+D_1} a\cdot d)\cdot (e^{\frac{1}{2}(D_2+D_3)-D_1} c\cdot b),
\label{id-gen-1}
\end{align}
where $D_i=\varepsilon_i D_x+\delta_i D_t$, $\varepsilon_i, \delta_i\in \mathbb{R}$,
and $a,b,c,d$ are  $C^{\infty}$ functions of  $(x,t)$.
\end{property}

It can be verified directly.

In the following we give some examples to show how the formula \eqref{id-gen-1} works in generating bilinear identities.

\vskip 6pt
\noindent
\textbf{Example} 4.1.1:  Taking  $D_2=D_3$ in \eqref{id-gen-1} yields
\begin{equation}
e^{D_1}(e^{D_2} a\cdot b)\cdot (e^{D_2} c\cdot d)
= (e^{D_2+D_1} a\cdot d)\cdot (e^{D_2-D_1} c\cdot b).
\label{id-e-1}
\end{equation}
Next, taking $D_1=\delta D_x,~ D_2=\varepsilon D_x$ in \eqref{id-e-1} and expanding the exponential functions of both sides, we have
\begin{align*}
 & (1+\delta D_x+\cdots) [(1+\varepsilon D_x+\cdots) a\cdot b)]\cdot [(1+\varepsilon D_x+\cdots) c\cdot d)]\\
=~& [(1+(\varepsilon+\delta) D_x+\frac{1}{2}(\varepsilon+\delta)^2 D^2_x +\cdots) a\cdot d)]\\
  &  ~~~~~ \times [(1+(\varepsilon-\delta) D_x+\frac{1}{2}(\varepsilon-\delta)^2 D^2_x +\cdots) c\cdot b)].
\label{id-gen-1}
\end{align*}
The coefficient of the term $\varepsilon \delta$ leads to a bilinear identity
\begin{equation}
D_x [(D_x a\cdot b)\cdot (cd)-(D_x c\cdot d)\cdot (ab)]
=(D^2_x  a\cdot d) bc- (D^2_x c\cdot b)ad.
\label{id-D-1}
\end{equation}

\vskip 6pt
\noindent
\textbf{Example} 4.1.2: Taking $D_2=D_3, b=c, d=a$ in \eqref{id-gen-1} yields
\begin{equation}
e^{D_1}(e^{D_2} a\cdot c)\cdot (e^{D_2} c\cdot a)
= (e^{D_2+D_1} a\cdot a)\cdot (e^{D_2-D_1} c\cdot c).
\label{id-e-2}
\end{equation}
Then we take $D_1=\varepsilon D_x,~ D_2=\delta D_t$, and from the coefficient of $\varepsilon \delta$ term in the
expansion we get
\begin{equation}
2D_x (D_t a\cdot c)\cdot (ac)=(D_xD_t  a\cdot a) c^2 - (D_xD_t c\cdot c)a^2;
\label{id-D-2}
\end{equation}
from $\varepsilon^4$ term we get
\begin{equation}
2D_x [(D_x^3 a\cdot c)\cdot (ac)-3(D_x^2 a\cdot c)\cdot (D_x a\cdot c)]=(D_x^4  a\cdot a) c^2 - (D_x^4 c\cdot c)a^2.
\label{id-D-3}
\end{equation}

In general, by specially taking $D_i$ in the formula \eqref{id-gen-1} and
comparing coefficients of $\{\varepsilon, \delta\}$, one can obtain various bilinear identities.
They play important roles in constructing bilinear BTs
and nonlinear superposition formulae based on the bilinear BTs.

\subsection{Bilinear B\"acklund transformations }\label{sec-4-2}

In this section we take the bilinear KdV equation \eqref{kdv-bil-D}, i.e.
\begin{equation}
(D_xD_t +D^4_x) f\cdot f=0,
\label{kdv-bil-D-f}
\end{equation}
as an example, to illustrate how a bilinear BT is constructed and how it works in finding solutions.

Suppose that $g$ is also a solution of \eqref{kdv-bil-D-f}, i.e.
\begin{equation}
(D_xD_t +D^4_x) g\cdot g=0.
\label{kdv-bil-D-g}
\end{equation}
Then the following relation holds
\begin{equation}
g^2(D_xD_t +D^4_x) f\cdot f- f^2 (D_xD_t +D^4_x) g\cdot g=0,
\label{BT-0-0}
\end{equation}
i.e.
\begin{equation}
[(D_xD_t f\cdot f) g^2- (D_xD_t g\cdot g) f^2]
+[(D^4_x f\cdot f)g^2- (D^4_x g\cdot g) f^2]=0.
\label{BT-0-1}
\end{equation}
Now, employing the identities \eqref{id-D-2} and \eqref{id-D-3}  (taking $a=f, c=g$), one can rewrite \eqref{BT-0-1} as
\begin{equation}
2D_x[(D_x^3+D_t )f\cdot g]\cdot (f g)
+ 6 D_x [(D_x f\cdot g)\cdot (D^2_x f\cdot g)]=0.
\label{BT-0-2}
\end{equation}
Next, introduce
\begin{subequations}\label{kdv-BT-bil}
\begin{equation}
D^2_xf\cdot g=\lambda fg,
\label{kdv-BT-bil-a}
\end{equation}
where $\lambda$ is a constant independent of $x$.
Substituting it into \eqref{BT-0-2} yields
\begin{equation*}
2D_x[(D_x^3+D_t +3\lambda D_x)f\cdot g]\cdot (f g)=0.
\end{equation*}
Then we can take
\begin{equation}
(D_x^3+D_t +3 \lambda D_x)f\cdot g =0.
\label{kdv-BT-bil-b}
\end{equation}
\end{subequations}
The two equations, \eqref{kdv-BT-bil-a} and \eqref{kdv-BT-bil-b},
compose a  coupled system which  provides a bilinear BT of the KdV equation \eqref{kdv-eq}.
In fact, from the above procedure we can see that
when \eqref{kdv-bil-D-f}, \eqref{kdv-BT-bil-a} and \eqref{kdv-BT-bil-b} hold, equation \eqref{BT-0-0} holds.
In other words,
if $f$ is a solution to \eqref{kdv-bil-D-f}, so is $g$ to \eqref{kdv-bil-D-g}, due to \eqref{BT-0-0}.
The above idea is due to Hirota \cite{Hir-1974}.

In many cases, using a bilinear BT to calculate soliton solutions
is not as convenient as directly using the original bilinear equation.
Now we show some details. First, taking $\lambda=\frac{k^2_1}{4}$ and $f=1$ (noticing that $f=1$ is a solution of \eqref{kdv-bil-D-f}),
and substituting them into \eqref{kdv-BT-bil}, one can find that $g$ needs to satisfy the following
linear equations:
\begin{align*}
& g_{xx}=\frac{k_1^2}{4}g,\\
& g_t + g_{xxx}+\frac{3}{4} k_1^2 g_x=0.
\end{align*}
As a solution we can take
\[g=g_1=e^{\frac{\eta_1}{2}}+ e^{-\frac{\eta_1}{2}},\]
where
\begin{equation}
\eta_i=k_i x-k_i^3 t+\eta_i^{(0)}, ~~ (k_i, \eta_i\in \mathbb{R}).
\label{eta-i-BT}
\end{equation}
Thus,  1SS of the KdV equation \eqref{kdv-eq} can be expressed by $u=2(\ln g_1)_{xx}$.

Next, taking $f=g_1,~\lambda= \frac{k^2_2}{4}$ and substituting them into \eqref{kdv-BT-bil}, we have
\begin{subequations}\label{kdv-BT-sol}
\begin{align}
& (D^2_x-\frac{k^2_2}{4})g\cdot (e^{\frac{\eta_1}{2}}+ e^{-\frac{\eta_1}{2}})=0,\label{kdv-BT-sol-a}\\
& (D_x^3+D_t +\frac{3}{4}k^2_2 D_x) g\cdot (e^{\frac{\eta_1}{2}}+ e^{-\frac{\eta_1}{2}})=0. \label{kdv-BT-sol-b}
\end{align}
\end{subequations}
To get a solution for $g$,
we assume $g$ to be the following form:
\[g=g_2=\alpha( e^{\frac{\eta_1+\eta_2}{2}}+ e^{-\frac{\eta_1+\eta_2}{2}})
+\beta( e^{\frac{\eta_1-\eta_2}{2}}+ e^{-\frac{\eta_1-\eta_2}{2}}),\]
where $\eta_i$ is defined as \eqref{eta-i-BT} and  $\alpha, \beta$ are undetermined constants.
Substituting the above $g$ into \eqref{kdv-BT-sol} we find: when $\alpha=k_1-k_2$ and $\beta=-(k_1+k_2)$, $g$ satisfies \eqref{kdv-BT-sol}.

Next, we can take $f=g_2,~\lambda= \frac{k^2_3}{4}$, and from \eqref{kdv-BT-bil} we solve out solution $g=g_3$.
However, it is apparent that the above procedure is not as ``mechanized" as the one we used in Sec.\ref{sec-2-2-1} to
calculate  NSS from the bilinear KdV equation \eqref{kdv-bil-D}.
In general we can successively take $\lambda= \frac{k^2_j}{4},~j=1,2,\cdots,N$ to calculate higher order solutions;
the general formula of $g$ for a generic $N$ has the following form:
\begin{equation}
g_N = \sum_{\varepsilon=\pm 1} \left[\, \prod_{1\leq j<l}^{N}\varepsilon_l(\varepsilon_j k_j-\varepsilon_l k_l)
e^{\frac{1}{2}\sum^N_{j=1}\varepsilon_j \eta_j} \right],
\label{kdv-Nss-BT}
\end{equation}
where $\eta_j$ is defined as \eqref{eta-i-BT}, the summation over $\varepsilon$ means to take all possible $\varepsilon_j=\{1,-1\}$ $(j=1,2,\cdots, N)$.

A proof that \eqref{kdv-Nss-BT} satisfies the bilinear KdV equation \eqref{kdv-bil-D}
can be found in Chapter 5 of \cite{Chen-book-2006} by making use of Wronskians.

\subsection{Deformation of bilinear BTs }\label{sec-4-3}

We have already seen that using bilinear BT \eqref{kdv-BT-bil}
to calculate soliton solutions is not as  convenient as using the original bilinear equation \eqref{kdv-bil-D}.
The reason is that $f=1,~g=1$ are not a solution pair to \eqref{kdv-BT-bil}.
To change such a situation, we try deforming bilinear BTs.

In \eqref{kdv-BT-bil} we replace $f$ and $g$ with $e^{\xi_1}f$ and $e^{\xi_2}g$, respectively, i.e.
\begin{equation}
f \to e^{\xi_1}f,~~~ g\to e^{\xi_2}g,
\end{equation}
where $\xi_i=p_ix+q_it+\xi^{(0)}_i,~ p_i,q_i,\xi^{(0)}_i\in \mathbb{R}$.
Noticing that solutions of the KdV equation can be expressed through  $u=2(\ln f)_{xx}$ or $u=2(\ln g)_{xx}$,
such a replacement does not change solutions; of course,
due to the gauge property \eqref{gauge-pro} of bilinear derivatives,
such a replacement does not change bilinear equations \eqref{kdv-bil-D-f} and \eqref{kdv-bil-D-g}  either.
Using formula \eqref{id-gauge-gen}, we have
\begin{align*}
& D_x(e^{\xi_1}f)\cdot (e^{\xi_2}g)=e^{\xi_1+\xi_2}[(p_1-p_2)fg+D_x f\cdot g],\\
& D_x^2(e^{\xi_1}f)\cdot (e^{\xi_2}g)=e^{\xi_1+\xi_2}[(p_1-p_2)^2fg+2(p_1-p_2)D_x f\cdot g+D_x^2 f\cdot g],\\
& D_x^3(e^{\xi_1}f)\cdot (e^{\xi_2}g)=e^{\xi_1+\xi_2}[(p_1-p_2)^3fg+3(p_1-p_2)^2 D_xf\cdot g\\
& ~~~~~~~~~~~~~~~~~~~~~~~~~~~~~~~~~~~~~~~~~~~~~~~~~~~~~~~~~+3(p_1-p_2)D_x^2 f\cdot g+D_x^3 f\cdot g],\\
& D_t(e^{\xi_1}f)\cdot (e^{\xi_2}g)=e^{\xi_1+\xi_2}[(q_1-q_2)fg+D_xf\cdot g].
\end{align*}
By means of them we rewrite \eqref{kdv-BT-bil} into
\begin{subequations}\label{kdv-BT-bil-D1}
\begin{align}
& [D_x^2+ 2(p_1-p_2)D_x ]f\cdot g= [\lambda -(p_1-p_2)^2]fg,\label{kdv-BT-bil-D1-a}\\
& \{D_t+D_x^3+ 3(p_1-p_2)D_x^2+3 [\lambda +(p_1-p_2)^2]D_x\} f\cdot g \nonumber \\
& ~~~~~~~~~~~~~~~~~~~~~~~~ = -[(q_1-q_2)+(p_1-p_2)^3+3\lambda (p_1-p_2)^2 ]fg. \label{kdv-BT-bil-D1-b}
\end{align}
\end{subequations}
Introducing
\[2(p_1-p_2)=\lambda',~~ \lambda=(p_1-p_2)^2,~~ (q_1-q_2)+4(p_1-p_2)^3=0,\]
so that we can simplify \eqref{kdv-BT-bil-D1-a} to
\begin{subequations}\label{kdv-BT-bil-D2}
\begin{align}
 (D_x^2+ \lambda' D_x )f \cdot g= 0,\label{kdv-BT-bil-D2-a}
\end{align}
and further, eliminating $D_x^2$ term in \eqref{kdv-BT-bil-D1-b} yields
\begin{align}
 (D_t+  D_x^3 )f \cdot g= 0.\label{kdv-BT-bil-D2-b}
\end{align}
\end{subequations}
Equations \eqref{kdv-BT-bil-D2-a} and \eqref{kdv-BT-bil-D2-b}
compose a deformed bilinear BT of the KdV equation \eqref{kdv-eq}.
Note that this deformed bilinear BT connects with a new discrete integrable equation
(cf.\cite{ChoMZ-2021}, equation (5.20)).

Compared with \eqref{kdv-BT-bil},
the deformed BT \eqref{kdv-BT-bil-D2} admits solutions $f=g=1$,
which facilitates the calculation: it allows to calculate $f$ and $g$ by  perturbation expansions.
Assuming
\begin{equation}
f=1+\sum^{\infty}_{i=1} f^{(i)} \varepsilon^i,~~~ g=1+\sum^{\infty}_{i=1} g^{(i)} \varepsilon^i
\label{fg-expand}
\end{equation}
and substituting them into \eqref{kdv-BT-bil-D2} yield
\begin{subequations}\label{fg-BT-exp}
\begin{align}
& (\partial_x^2+\lambda'\partial_x)(f^{(1)}-g^{(1)})=0,\label{fg-BT-exp-a}\\
& (\partial_x^2+\lambda'\partial_x)(f^{(2)}-g^{(2)})
   = - (D_x^2+ \lambda' D_x )f^{(1)} \cdot g^{(1)},\label{fg-BT-exp-b}\\
& (\partial_x^2+\lambda'\partial_x)(f^{(3)}-g^{(3)})
   =- (D_x^2+ \lambda' D_x )(f^{(1)} \cdot g^{(2)}+f^{(2)} \cdot g^{(1)}),\label{fg-BT-exp-c}\\
& \cdots \cdots; \nonumber
\\
& (\partial_t+\partial_x^3)(f^{(1)}-g^{(1)})=0,\label{fg-BT-exp-d}\\
& (\partial_t+\partial_x^3)(f^{(2)}-g^{(2)})
   = - (D_t+  D_x^3 )f^{(1)} \cdot g^{(1)},\label{fg-BT-exp-e}\\
& (\partial_t+\partial_x^3)(f^{(3)}-g^{(3)})
   =- (D_t+  D_x^3 )(f^{(1)} \cdot g^{(2)}+f^{(2)} \cdot g^{(1)}),\label{fg-BT-exp-f}\\
& \cdots \cdots. \nonumber
\end{align}
\end{subequations}
To get solutions, first, we can take $g^{(j)}=0,~ (j\geq 1)$, which corresponds to a zero solution
to the KdV equation \eqref{kdv-eq}.
Taking $\lambda'=-k_1$, from \eqref{fg-BT-exp-a} and \eqref{fg-BT-exp-b} we find a solution
\begin{equation}
f^{(1)}=e^{\eta_1},~~ \eta_i=k_i x-k_i^3 t+\eta^{(0)}_i,~ ~ (k_i,\eta^{(0)}_i\in \mathbb{R}),
\label{f1-eta-BT}
\end{equation}
where note that  $\eta_i$ is as same  as \eqref{eta-i-kdv}.
For those $f^{(j)},~(j\geq 2)$ we can trivially take them to be 0.
Thus, a 1SS is obtained by $u=2[\ln(1+f^{(1)})]_{xx}$.

Next, still considering $g$ as a new seed solution and taking $g^{(1)}=e^{\eta_1}$, $g^{(j)}=0,~ (j\geq 2)$,
in the following we will see that the deformed bilinear BT does bring some new aspects.

\vskip 7pt
\noindent
{Case 1}: Taking $\lambda'=-k_2$ and assuming
\begin{equation*}
f^{(1)}=a e^{\eta_1}+be^{\eta_2},
\end{equation*}
where $\eta_i$ is defined in \eqref{f1-eta-BT},
from \eqref{fg-BT-exp-a} and \eqref{fg-BT-exp-b}
we find  $a=-\frac{k_1+k_2}{k_1-k_2}$ and $b$ can be an arbitrary constant.
Next, from \eqref{fg-BT-exp-b} and \eqref{fg-BT-exp-e}   we find
\begin{equation*}
f^{(2)}=-\frac{b}{a}  e^{\eta_1+\eta_2}.
\end{equation*}
For $f^{(j)},~ (j\geq 3)$, we can take them to be 0.
Thus, a 2SS is obtained by $u=2(\ln f)_{xx}$, where
\begin{equation}
f= 1+a  e^{\eta_1}+be^{\eta_2}-\frac{b}{a}  e^{\eta_1+\eta_2},~~ a=-\frac{k_1+k_2}{k_1-k_2}.
\label{2ss-f-DBT}
\end{equation}

\vskip 7pt
\noindent
{Case 2}: Still take $\lambda'=-k_1$ (as in getting 1SS). In this case we have
\begin{equation*}
f^{(1)}=\zeta_1 e^{\eta_1},~~~ f^{(2)}= e^{2\eta_1},
\end{equation*}
where $\eta_1$ is defined in \eqref{f1-eta-BT}, $\zeta_1=2k_1(x-3k_1^2 t)+\zeta^{(0)},~ \zeta^{(0)}\in \mathbb{R}$,
and $f^{(j)}=0,~ (j\geq 3)$.
The obtained solution is
$u=2(\ln f)_{xx}$, where
\begin{equation}
f= 1-\zeta_1 e^{\eta_1}+ e^{2\eta_1}.
\label{2ps-f-DBT}
\end{equation}
The solution is known as a double-pole solution (corresponding to the case $k_2\to k_1$).

\vskip 7pt
The above are two examples. If in \eqref{2ss-f-DBT} we scale $a,b$ to be 1 by redefining the constants $\eta^{(0)}_i$, $f$ can be written as
\begin{equation*}
f=1+ e^{\eta_1} + e^{\eta_2}+ A_{12}\, e^{\eta_1+\eta_2},~~ A_{12}=\Bigl(\frac{k_1-k_2}{k_1+k_2}\Bigr)^2,
\end{equation*}
which is the standard Hirota's form for 2SS of the KdV equation.
If in \eqref{2ss-f-DBT} we take $b=-a$ and take the limit $k_2\to k_1$, we can obtain \eqref{2ps-f-DBT}.
Double-pole solutions correspond to the case that the transmission coefficient $T(k)$
in the Inverse Scattering Transform (IST) has
a double pole (note that simple poles lead to solitons).
Double-pole solutions can be obtained from many different ways,
for example,  from superposition formula \cite{Lamb-1971},
IST \cite{ZakS-1972,WadO-1982}, Darboux transformation \cite{Mat-PLA-1992-b,Mat-PLA-1992-b},
Wronskian technique \cite{SirHR-1988,Zha-2006},
bilinear method \cite{CheZD-2002,DenC-2001},
Cauchy matrix approach \cite{XuZZ-2014,ZhaZ-SAPM-2013}, etc.

The deformed bilinear BTs have advantages in allowing more freedom in calculations and
obtaining more kinds of solutions.
If we keep taking  $\lambda'=0$ in  each step of the transformations, we may obtain high order rational solutions.
Besides, rational solutions can also be derived directly from bilinear equations  \cite{Hir-1980},
or from superposition formulae \cite{AblS-1978},
or from determinantal approach (e.g. \cite{Zha-2006}).
For more examples on the deformed bilinear BTs, one can refer to
\cite{ChenBC-2004,BiCC-2014}, etc.

\subsection{BTs and Lax pairs}\label{sec-4-4}

A bilinear BT appears as an system of equations and actually requires compatibility among these equations.
This fact brings BTs and Lax pairs together.

For the KdV equation \eqref{kdv-eq}, i.e.
\begin{equation}
u_t+6uu_x+u_{xxx}=0,
\label{kdv-eq-2}
\end{equation}
its Lax pair reads  (e.g.\cite{AblS-1981})
\begin{subequations}\label{kdv-lax}
\begin{align}
&\phi_{xx}+u\phi=\lambda \phi, \label{kdv-lax-a}
\\
& \phi_{t}=\phi_{xxx}+3(\lambda+u)\phi_x,
\label{kdv-lax-b}
\end{align}
\end{subequations}
where $\phi$ serves as the eigenfunction and $\lambda$ as the spectral parameter.
Noticing that in the bilinear BT \eqref{kdv-BT-bil}, $f$ and $g$ correspond to two solutions of the KdV equation:
$u=2(\ln f)_{xx}$ and $\t u=2(\ln g)_{xx}$,
we put them together and we have
\begin{equation}
\t u =u+2(\ln\phi)_{xx},~~~\phi=\frac{g}{f},
\end{equation}
which is as same as the relation of two solutions of the KdV equation obtained from the Darboux transformation \cite{MatS-1991}.

In the bilinear BT \eqref{kdv-BT-bil}, taking
\begin{equation}
u=2(\ln f )_{xx},~~~\phi=\frac{g}{f}
\label{ufgphi}
\end{equation}
and rewriting \eqref{kdv-BT-bil} in terms of  $u$ and $\phi$ directly yield the Lax pair \eqref{kdv-lax}.
In the reverse direction,   substituting \eqref{ufgphi} into the Lax pair \eqref{kdv-lax}, one can find
the bilinear BT \eqref{kdv-BT-bil}.

In general, once one has a bilinear equation, one can construct its bilinear BT,
and with suitable substitutions from the bilinear  BT one can obtain a Lax pair for
the related nonlinear equation.
Hirota extended this idea to the discrete case and found bilinear BTs and Lax pairs
for new discrete integrable equations,
e.g. \cite{Hir-1977a}.

There is a nonlinear BT which was presented by  Wahlquist and Estabrook in 1973  \cite{WahE-1973}:
\begin{subequations}\label{kdv-BT-non-1}
\begin{align}
& (\t w+w)_x=2\lambda-\frac{1}{2}(\t w -w)^2, \label{kdv-BT-non-1-a}\\
& (\t w-w)_t=\frac{1}{2}[(\t w-w)^3]_x-6\lambda (\t w-w)_{x}-(\t w-w)_{xxx}, \label{kdv-BT-non-1-b}
\end{align}
\end{subequations}
where  $w$ satisfies the potential KdV equation \eqref{kdv-eq-poten}, i.e. $u=w_x$ satisfying the KdV equation \eqref{kdv-eq}.
Once $\t w$ is solved out from \eqref{kdv-BT-non-1},  $\t u=\t w_x$ provides a new solution to the KdV equation.

The BT \eqref{kdv-BT-non-1} can be bilinearized  \cite{Hir-1980}.
Taking
\begin{equation}
w=2(\ln f )_{x},~~\t w =2(\ln g )_{x},
\label{fgwtw}
\end{equation}
and substituting them into \eqref{kdv-BT-non-1}, one can obtain the bilinear BT \eqref{kdv-BT-bil}.

With \eqref{fgwtw} and $\phi=g/f$, which yields $\t w-w=2(\ln \phi)_x$, from \eqref{kdv-BT-non-1}
one can directly obtain the Lax pair \eqref{kdv-lax} of the KdV equation.

\vskip 8pt
\textit{Note: Lax pair, bilinear BT and nonlinear BT
provide different forms for the same thing.
Around the year 1974,
many researchers considered relations between BTs and Lax pairs,
see, e.g.\cite{Che-1974,Hir-1974,Lamb-1974,WadSK-1975,Miu-1976}, etc.
}

\subsection{BTs and superposition formulae}\label{sec-4-5}

\subsubsection{Nonlinear BTs and superposition formulae}\label{sec-4-5-1}

How can we use a BT to generate solutions? For the nonlinear BT \eqref{kdv-BT-non-1},
taking $w=0$  and $\lambda=k_1^2$, we have
\begin{subequations}\label{kdv-BT-non-2}
\begin{align}
& \t w_x=2k_1^2-\frac{1}{2}\t w^2, \label{kdv-BT-non-2-a}
\\
& \t w_t=\frac{1}{2}(\t w^3)_x-6k_1^2\t w_{x}-\t w_{xxx}.\label{kdv-BT-non-2-b}
\end{align}
\end{subequations}
From \eqref{kdv-BT-non-2-a} we can assume
\[\t w=2k_1 \tanh (kx+c(t)),\]
where $c(t)$ is undetermined. Substituting it into \eqref{kdv-BT-non-2-b} we find
$c(t)=4k^3_1 t+\eta^{(0)}_1$, and thus,
\begin{equation}
\t w=2k_1 \tanh \eta_1,~~~ \eta_i=k_i x+4k_i^3 t +\eta^{(0)}_i,~~ (k_i, \eta_i^{(0)}\in \mathbb{R})
\label{tw}
\end{equation}
provides a solution to the potential KdV equation \eqref{kdv-eq-poten} and $u=\t w_x$ is a 1SS of the KdV equation \eqref{kdv-eq}.

Next, taking $w$ to be \eqref{tw} and substituting it into the BT \eqref{kdv-BT-non-1} to solve $\t w$,
in principle, solving  BT \eqref{kdv-BT-non-1} with $\lambda=k_2^2$,
one will get a 2SS for the KdV equation.
Apparently, this becomes complicated in calculation and is not as convenient as using
the bilinear BT.
However, once we get the first few solutions from the B\"acklund transformation \eqref{kdv-BT-non-1},
a recursive relation of these solutions can be built, from which  new solutions can easily be generated.
Such a recursive relation is referred to as a nonlinear superposition formula of solutions.
In the next, let us derive the nonlinear superposition formula for the  KdV equation.
We only use equation \eqref{kdv-BT-non-1-a}
in the BT.

We start from \eqref{kdv-BT-non-1-a} with a seed solution $w$,
denote $\t w=w_1$ when taking $\lambda=\lambda_1$
and denote $\t w=w_2$ when taking $\lambda=\lambda_2$, i.e.
\begin{subequations}\label{kdv-BT-non-1-aa}
\begin{align}
& (w_1+w)_x=2\lambda_1-\frac{1}{2}(w_1 -w)^2, \label{kdv-BT-non-1-aa-1}\\
& (w_2+w)_x=2\lambda_2-\frac{1}{2}(w_2 -w)^2. \label{kdv-BT-non-1-aa-2}
\end{align}
\end{subequations}
Then, using \eqref{kdv-BT-non-1-a} with a new seed $w=w_1$ and $\lambda=\lambda_2$,
and denoting the obtained solution  by $\t w=w_{12}$, we have
\begin{subequations}\label{kdv-BT-non-1-ab}
\begin{align}
(w_{12}+w_1)_x=2\lambda_2-\frac{1}{2}(w_{12} -w_1)^2; \label{kdv-BT-non-1-ab-1}
\end{align}
and meanwhile, taking $w=w_2, ~\lambda=\lambda_1$ in \eqref{kdv-BT-non-1-a} and denoting $\t w=w_{21}$, yield
\begin{align}
(w_{21}+w_2)_x=2\lambda_1-\frac{1}{2}(w_{21} -w_2)^2. \label{kdv-BT-non-1-ab-2}
\end{align}
\end{subequations}
The above procedure can be described as in Figure \ref{fig-2-1}.
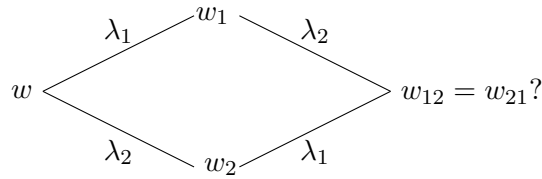
\begin{figure}[!ht]
\vskip 0.5cm
\begin{center}
    \begin{tikzpicture}
      \draw (-2, 0) node [left]{$w$} -- node [above]{$\lambda_1$} (0, 1);
       \draw (0.6, 1) node [left]{$w_1$} -- node [above]{$\lambda_2$} (2.6, 0) node [right]{$w_{12}=w_{21}?$};
       \draw (-2, 0) node [left]{$~$}-- node [below]{$\lambda_2$} (0, -1) node [right]{$w_2$};
       \draw (0.6, -1) node [above]{$~$} --  node [below]{$\lambda_1$} (2.6, 0) node [right]{$~$};
    \end{tikzpicture}
\end{center}
\vskip -0.5cm
\caption{Permutability  property of solutions based on  B\"acklund transformation.}\label{fig-2-1}
\end{figure}
The question is whether $w_{12}$ and $w_{21}$ are same.
To answer this question, we eliminate $w_{1,x}$ from \eqref{kdv-BT-non-1-aa-1} and \eqref{kdv-BT-non-1-ab-1}
and we arrive at
\[w_1=\frac{1}{2}(w_{12}+w)+\frac{2(\lambda_1-\lambda_2)}{w_{12}-w}+[\ln(w_{12}-w)]_x.\]
Substituting it into \eqref{kdv-BT-non-1-ab-1} we find
\begin{align}
\lambda_1+\lambda_2 = & (w_{12}+w)_x + [\ln(w_{12}-w)]_{xx}+\frac{1}{2}[\ln(w_{12}-w)]_x^2\nonumber\\
                      & ~~~ + \frac{1}{8}(w_{12}-w)^2 + \frac{2(\lambda_1-\lambda_2)^2}{(w_{12}-w)^2}.\label{w12-ode}
\end{align}
This can be viewed as an ordinary differential equation (ODE) for both $w_{12}$ and $w_{21}$,
because it is invariant if switching $\lambda_1$ and $\lambda_2$ in the equation.
Thus, once we impose a same initial condition on $w_{12}$ and $w_{21}$,
we will get  $w_{12}=w_{21}$.

To obtain a neat form of the recursive relation for the solutions, from \eqref{kdv-BT-non-1-aa-1} and \eqref{kdv-BT-non-1-aa-2} we have
\[(w_1-w_2)_x=2(\lambda_1-\lambda_2)-\frac{1}{2}(w_1-w_2)(w_1+w_2-2w);\]
and from \eqref{kdv-BT-non-1-ab-1} and \eqref{kdv-BT-non-1-ab-2} we have (noticing that $w_{12}=w_{21}$)
\[(w_1-w_2)_x=-2(\lambda_1-\lambda_2)-\frac{1}{2}(w_1-w_2)(w_1+w_2-2w_{12}).\]
Eliminating the derivative terms from them yields a purely algebraic expression:
\begin{equation}
4(\lambda_1-\lambda_2)=(w_1-w_2)(w_{12}-w),
\label{kdv-nsf}
\end{equation}
which is referred to as the nonlinear superposition formula of solutions of the (potential) KdV equation,
also known as the Bianchi identity\footnote{
It is Bianchi who first derived a nonlinear superposition formula of solutions of the sG equation
and first proved permutation property of solutions \cite{Bia-1899}.}, also called the discrete potential KdV equation \cite{NijCW-1985,NijC-1995}.
As a discrete equation, \eqref{kdv-nsf} is usually written as
\begin{equation}
(w_{n+1,m}-w_{n,m+1})(w_{n,m}-w_{n+1,m+1})=q^2-p^2,
\label{kdv-dis}
\end{equation}
in which $w_{n,m}=w(n,m)$, $(n,m)\in \mathbb{Z}^2$,
$p$ and $q$ are  spacing parameters of $n$- and $m$- direction, respectively.

It is remarkable that the derivation of the nonlinear superposition formula is only based on $x$-part in the
BT, which implies the superposition formula is valid for the solutions of the whole KdV hierarchy.
The formula is so simple and neat.
In addition to the KdV equation, some equations such as  the  mKdV equation, sG equation,
mKdV-sG equation  and so on,
have superposition formulae in simple forms
(e.g. \cite{Lamb-1971,Lamb-1974,Che-1974,Wad-1974,KonS-1975,Kon-1982}).
In fact, the mKdV equation, sG equation and mKdV-sG equation share a same superposition formula,
cf.\cite{ChoWZ-2023}.
Besides, when these nonlinear superposition formulae are treated as 2-dimensional discrete equations,
they exhibit beautiful 3-dimensional consistency \cite{HieJN-2016}.

\textit{Note: Here are some remarks on the B\"acklund transformation and Bianchi identity.
For more details one can refer to the book \cite{RogS-2002} by Rogers and Schief
and the paper \cite{PruS-1998} by  Prus and Sym.
Historically, Bianchi (in his habilitation thesis in 1879)
introduced a purely geometric construction for pseudospherical surfaces.
He  reformulated the construction
in mathematical terms as a transformation $\mathbb{B}$  (without any parameter).
Later, B\"acklund (1883) published  his celebrated transformation $\mathbb{B}_{\sigma}$
($\sigma$ is a parameter):
\begin{subequations}
\label{BT-sigma}
\begin{align}
& \phi_{u}+\theta_{v}=(\sin \phi \cos \theta+\sin \sigma \cos \phi \sin \theta)/\cos\sigma,\\
& \phi_{v}+\theta_{u}=(-\cos \phi \sin \theta+ \sin \sigma \sin \phi \cos \theta)/\cos\sigma,
\end{align}
\end{subequations}
where $\theta=\theta(u,v),~\phi=\phi(u,v)$ and  $u,v$ are independent variables.
It
extends Bianchi's construction and  allows
the iterative construction of pseudospherical surfaces.
Bianchi (1885) showed the above transformation to be associated with an elegant invariance of the sG equation:
\[\theta_{uu}-\theta_{vv}=\sin \theta \cos \theta,\]
if $\theta$ is a solution and $\phi$ satisfies \eqref{BT-sigma},
then $\phi$ is a solution of the sG equation
\(\phi_{uu}-\phi_{vv}=\sin \phi \cos \phi.\)
Before Bianchi, Darboux (1883) already found the above  invariance for a special case $\sigma=0$.
Later,
Bianchi (1892) demonstrated that  B\"acklund's transformation $\mathbb{B}_{\sigma}$
admits a commutativity  property
$\mathbb{B}_{\sigma_1}\mathbb{B}_{\sigma_2}=\mathbb{B}_{\sigma_2}\mathbb{B}_{\sigma_1}$,
which is termed as the ``Permutability Theorem''.
Based on the permutability, he was able to  obtained a nonlinear superposition formula (Bianchi identity)
of B\"acklund's transformation.
}

\subsubsection{Bilinear BTs and superposition formulae}\label{sec-4-5-2}

Using bilinear BTs one can derive bilinear superposition formulas for bilinear equations,
which was first found by Hirota and Satsuma in 1978 \cite{HirS-1978}.
Still we take the KdV equation as an example to show the construction.

We start from the $x$-part of the bilinear BT \eqref{kdv-BT-bil} of the KdV equation, i.e.
\begin{equation}
D^2_x f\cdot \t f=\lambda f \t f,
\label{kdv-BT-bil-x}
\end{equation}
where for convenience we have replaced $g$ with $\t f$.
We will construct a relation as shown in Figure \ref{fig-2-2}.
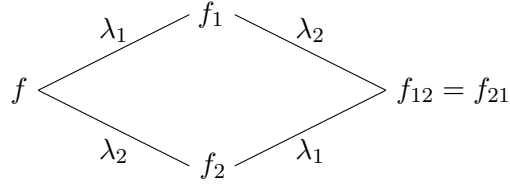
\begin{figure}[!ht]
\vskip 0.5cm
\begin{center}
    \begin{tikzpicture}
      \draw (-2, 0) node [left]{$f$} -- node [above]{$\lambda_1$} (0, 1);
       \draw (0.6, 1) node [left]{$f_1$} -- node [above]{$\lambda_2$} (2.6, 0) node [right]{$f_{12}=f_{21}$};
       \draw (-2, 0) node [left]{$~$}-- node [below]{$\lambda_2$} (0, -1) node [right]{$f_2$};
       \draw (0.6, -1) node [above]{$~$} --  node [below]{$\lambda_1$} (2.6, 0) node [right]{$~$};
    \end{tikzpicture}
\end{center}
\vskip -0.5cm
\caption{Permutability property of  bilinear B\"acklund transformations.}\label{fig-2-2}
\end{figure}
Of course, first we need to investigate the possibility $f_{12}=f_{21}$. Similar to Sec.\ref{sec-4-5-1},
from \eqref{kdv-BT-bil-x} we have
\begin{subequations}\label{f12bt}
\begin{align}
& (D_x^2-\lambda_1)f \cdot f_{1}=0,\label{f12bt-a}\\
& (D_x^2-\lambda_2)f \cdot f_{2}=0,\label{f12bt-b}\\
& (D_x^2-\lambda_2)f_1 \cdot f_{12}=0,\label{f12bt-c}\\
& (D_x^2-\lambda_1)f_2 \cdot f_{21}=0.\label{f12bt-d}
\end{align}
\end{subequations}
Introducing $w=2(\ln f)_x$, we may write \eqref{f12bt-a} as
\[(w+w_1)_x=2\lambda_1-\frac{1}{2}(w-w_1)^2,\]
which is the same as \eqref{kdv-BT-non-1-aa-1}.
Equations \eqref{f12bt-b}, \eqref{f12bt-c} and \eqref{f12bt-d} can also be written as  \eqref{kdv-BT-non-1-aa-2}, \eqref{kdv-BT-non-1-ab-1} and \eqref{kdv-BT-non-1-ab-2} in terms of $w$.\footnote{This shows
the coincidence between bilinear and nonlinear BTs and Lax pairs.}
Thus, both $w_{12}=2(\ln f_{12})_{x}$ and  $w_{21}=2(\ln f_{21})_{x}$ will satisfy the same ODE \eqref{w12-ode},
and $f_{12}$ and $f_{21}$ can be same provided they enjoy same initial conditions.

Now, $f_2f_{12}\times \eqref{f12bt-a} - ff_1\times \eqref{f12bt-d}$ yields
\begin{equation}
(D^2_x f\cdot f_{1})f_2f_{12}-(D^2_x f_2\cdot f_{12})f f_{1} =0;
\label{f12D1}
\end{equation}
meanwhile, using bilinear identity \eqref{id-D-1} we have
\begin{equation}
(D^2_x f\cdot f_{1})f_2f_{12}-(D^2_x f_2\cdot f_{12})f f_{1} =D_x[(D_x f\cdot f_{12})\cdot(f_1 f_2)-D_x(f_2\cdot f_1)\cdot (ff_{12})].
\label{f12D2}
\end{equation}
By a comparison we immediately find
\begin{equation*}
D_x[(D_x f\cdot f_{12})\cdot(f_1 f_2)-(D_x f_2\cdot f_1)\cdot (ff_{12})]=0.
\end{equation*}
Switching indices: $1\leftrightarrow 2$, yields
\begin{equation*}
D_x[(D_x f\cdot f_{12})\cdot(f_1 f_2)-(D_x f_1\cdot f_2)\cdot (ff_{12})]=0.
\end{equation*}
For the above two equations, adding and subtracting each other yield, respectively
\begin{align*}
D_x(D_x f\cdot f_{12})\cdot(f_1 f_2)=0,\\
D_x(D_x f_1\cdot f_2)\cdot (ff_{12})=0.
\end{align*}
Then, noticing the property $D_x g\cdot g=0$ we can take
\begin{subequations}\label{kdv-nsf-bil}
\begin{align}
D_x f\cdot f_{12}=\alpha f_1 f_2,\\
D_x f_2\cdot f_1 =\beta ff_{12},
\end{align}
\end{subequations}
where $\alpha, \beta\in \mathbb{R}$.
These two equations together compose a bilinear superposition formula for the bilinear
KdV equation \eqref{kdv-bil-D}.

One can derive the  superposition formula \eqref{kdv-nsf} from \eqref{kdv-nsf-bil}.
In fact, multiplying each other in \eqref{kdv-nsf-bil} yields
\[ (D_x f\cdot f_{12})\times (D_x f_2\cdot f_1)=\alpha \beta ff_1 f_2f_{12}.
\]
Then introducing $w=2(\ln f)_x$, we arrive at
\begin{equation}
(w-w_{12})(w_1-w_2)=-4 \alpha \beta,
\end{equation}
which is the nonlinear superposition formula \eqref{kdv-nsf}.
Thus, two forms of the superposition formulas are unified.

\section{Vertex operators and tau functions} \label{sec-5}

Vertex operator plays a key role in the theory developed by the Kyoto group \cite{MiwJD-2000}.
In this section, we only introduce the transformations between $\tau$ functions
manipulated by vertex operators.

\subsection{Vertex operator of the KdV equation}\label{sec-5-1}

The function $f$ defined by \eqref{Nss-f} corresponds to the NSS of the KdV equation \eqref{kdv-eq}.
It is called the $\tau$ function of the KdV equation, denoted by $\tau_N$.
The bilinear BT \eqref{kdv-BT-bil} provides a transformation between $\tau_N$ and $\tau_{N+1}$,
which can be proved in an apparent way when the $\tau$ functions are expressed in Wronskians
\cite{NimF-JPA-1984-BT}.
When $\tau_N$ is presented in Hirota's form \eqref{Nss-f},
the transformation can be described by means  an operator $X(k)$ (called vertex operator),
which provides a transformation directly by:
\begin{equation}
\tau_{N+1}=e^{c_{N+1}X(k_{N+1})} \tau_N.
\end{equation}
We will explain such an operator in this subsection.

We rewrite the KdV equation (by $t\to -4t$) as
\begin{equation}
4u_t-6uu_x-u_{xxx}=0.
\label{kdv-eq-sato}
\end{equation}
Then, under the transformation
\begin{equation}
u=2(\ln \tau)_{xx},
\label{tran-kdv-tau}
\end{equation}
the bilinear KdV equation is
\begin{equation}
(4D_xD_t -D^4_x) \tau\cdot \tau=0.
\label{kdv-bil-D-tau}
\end{equation}
For convenience we introduce $t_1=x,~ t_3=t$, with which the above bilinear equation reads
\begin{equation}
(4D_1D_3 -D^4_1) \tau\cdot \tau=0.
\label{kdv-bil-D-tau-13}
\end{equation}
Its NSS is given by
\begin{subequations}\label{Nss-tau}
\begin{equation}
\tau_N =\sum_{\mu=0,1} \mathrm{exp}\left(2 \sum^{N}_{j=1} \mu_j (\zeta_j+\zeta^{(0)}_j)+\sum^N_{1\leq i<j}\mu_i\mu_j a_{ij}\right),
\label{Nss-tau-a}
\end{equation}
where
\begin{equation}
\zeta_j=\sum_{i=0}^{\infty} k_j^{2i+1}t_{2i+1}, ~~ e^{a_{ij}}=A_{ij}=\Bigl(\frac{k_i-k_j}{k_i+k_j}\Bigr)^2,
\label{Nss-tau-b}
\end{equation}
\end{subequations}
$\zeta^{(0)}_j\in \mathbb{R}$   and the summation over $\mu$ is as same as in \eqref{Nss-f}.
Here we note that, to employ those notations in Sato's theory, we use infinite coordinates $(t_1,t_3,t_5,\cdots)$.
In fact, for the KdV equation we can treat $(t_5,t_7,\cdots)$ (which do not appear in the equation) as parameters.
This is guaranteed by the infinitely many isospectral symmetries of the KdV equation (see \cite{MiwJD-2000}).

The above $\tau_N$ can be rewritten as  \cite{DatKM-1981}
\begin{equation}
\tau_N=\sum_{J\subset I}\Biggl[ \Biggl(\prod_{i\in J}c_i\Biggr)\Biggl(\prod_{\begin{smallmatrix}i,j\in J\\i<j\end{smallmatrix}}A_{ij}\Biggr)
\mathrm{exp}\Biggl(2\sum_{i\in J} \zeta_i\Biggr)\Biggr],
\label{Nss-tau-sato}
\end{equation}
where  $c_i\in \mathbb{R}$, $I$ stands for the set $I=\{1,2,\cdots,N\}$, $J$ is a subset of $I$,
and summation over $J\subset I$ means taking all possible subsets of $I$.
In the above expression, $J=\varnothing$ corresponds to ``$1$",
$J=\{i\}$ corresponds to  $c_i e^{2\zeta_i}$,
and $J=\{1, 2\}$ corresponds to $c_1c_2 A_{12} e^{2(\zeta_1+\zeta_2)}$, $\cdots\cdots$.
Obviously, $c_i$ takes the place of $e^{2\zeta^{(0)}_i}$ in  \eqref{Nss-tau}.
When $N=2$, we have
\[\tau_2=1+c_1e^{2\zeta_1} + c_2e^{2\zeta_2}+ c_1 c_2 A_{12}\, e^{2(\zeta_1+\zeta_2)},
\]
which is the same as \eqref{2ss-f} $(\varepsilon=1)$.

For the transformations between $\tau$ functions based on vertex operators, there is the following theorem \cite{DatKM-1981,MiwJD-2000}.

\begin{theorem}\label{thm-2-6-1}
For the $\tau$ function defined by \eqref{Nss-tau-sato}, there is
\begin{equation}
\tau_{N+1}=e^{c_{N+1}X(k_{N+1})} \tau_N,
\label{tau-X}
\end{equation}
where
\begin{subequations}\label{X-kdv}
\begin{align}
& X(k)=e^{2\zeta(\mathbf{t},k)} \, e^{-2\zeta(\t \partial,\, k^{-1})},\label{X-kdv-a}\\
& \zeta(\mathbf{t},k)=\sum^{\infty}_{j=0}k^{2j+1} t_{2j+1},~~~ \mathbf{t}=(t_1, t_3, t_5, \cdots),\label{X-kdv-b}\\
& \t{\partial}=\Bigl(\partial_{1},\, \frac{\partial_{3}}{3},\, \frac{\partial_{5}}{5},\, \cdots \Bigr),~~~ \partial_{j}=\partial_{t_j}. \label{X-kdv-c}
\end{align}
\end{subequations}
\eqref{X-kdv-a} is called a vertex operator.
\end{theorem}


We prove this theorem through some lemmas.

\begin{lem}\label{lem-2-6-1}
$\forall\, a,k\in \mathbb{R}$, there is
\begin{equation}
 e^{a\zeta(\t \partial,\, k^{-1})} f(\mathbf{t})=f(\mathbf{t}+a\varepsilon(k)),
\end{equation}
where
\begin{equation*}
\varepsilon(k)=\Bigl(\frac{1}{k},\, \frac{1}{3k^3},\, \frac{1}{5k^5},\, \cdots\Bigl).
\end{equation*}
\end{lem}

The result is obvious.

\begin{lem}\label{lem-2-6-2}
For $\zeta(\mathbf{t},k)$ defined in \eqref{X-kdv-b}, there is
\begin{equation}
e^{-4\zeta(\varepsilon(p), q)}=\Bigl(\frac{p-q}{p+q}\Bigr)^2.
\label{Apq}
\end{equation}
\end{lem}

\begin{proof}
Through formal expansion (assuming $p>q>0$) we have
\begin{align*}
\ln \Bigl(\frac{p-q}{p+q}\Bigr)
& = \ln (1-q/p)-\ln (1+q/p)\\
& = -\sum_{j=1}^{\infty}\frac{q^j}{j\, p^j}- \sum_{j=1}^{\infty}(-1)^{j+1}\frac{q^j}{j\, p^j}\\
& = -2 \sum_{j=0}^{\infty}\frac{q^{2j+1}}{(2j+1)\, p^{2j+1}}\\
& = -2 \zeta(\varepsilon(p), q),
\end{align*}
which is equivalent to \eqref{Apq}.
\end{proof}

\begin{lem}\label{lem-2-6-3}
For $\zeta(\mathbf{t},k)$ and $\t \partial$ defined in \eqref{X-kdv} and $A_{ij}$ defined in \eqref{Nss-tau-b}, there is
\begin{equation}
 e^{-2\zeta(\t \partial,\, k^{-1}_i)} e^{2\zeta(\mathbf{t},k_j)}
 =A_{ij}  e^{2\zeta(\mathbf{t},k_j)} \, e^{-2\zeta(\t \partial,\, k^{-1}_i)}.
\label{switch}
\end{equation}
\end{lem}
\begin{proof}
For any $C^\infty$ function $f(\mathbf{t})$, successively using Lemma \ref{lem-2-6-1} and Lemma \ref{lem-2-6-2}, we find
\begin{align*}
& e^{-2\zeta(\t \partial,\, k^{-1}_i)} e^{2\zeta(\mathbf{t},k_j)}  f(\mathbf{t})\\
=~& e^{2\zeta(\mathbf{t}-2\varepsilon(k_i) ,k_j)}  f(\mathbf{t}-2\varepsilon(k_i))\\
=~& e^{2\zeta(\mathbf{t},k_j)} e^{-4\zeta(\varepsilon(k_i) ,k_j)} e^{-2\zeta(\t \partial,\, k^{-1}_i)}  f(\mathbf{t})\\
=~& A_{ij}  e^{2\zeta(\mathbf{t},k_j)} \, e^{-2\zeta(\t \partial,\, k^{-1}_i)} f(\mathbf{t}).
\end{align*}
The lemma is proved due to the arbitrariness of $f(\mathbf{t})$.
\end{proof}

There is another expression for \eqref{switch}. If denoting  $A=-2\zeta(\t \partial,\, k^{-1}_i),~ B=2\zeta(\mathbf{t},k_j)$,
first we have
\begin{equation}
[A,B]=- 4\zeta(\varepsilon(k_i) ,k_j)=\ln  A_{ij},
\label{AB-BA}
\end{equation}
where $[A,B]=AB-BA$. In fact, because $B$ is a linear function, $[A,B]$ must be a scalar.
Noticing that
\[[\partial_{2r+1},\, t_{2s+1}]=\delta_{r,s},\]
where $\delta_{r,s}$ is the Kronecker $\delta$ function,
we have
\begin{align*}
[A,B]&=-4\sum^{\infty}_{r=0}\sum^{\infty}_{s=0}\frac{1}{2s+1}\frac{k_j^{2r+1}}{k_i^{2s+1}}[\partial_{2s+1},\, t_{2r+1}]\\
     &=-4 \sum^{\infty}_{r=0}\sum^{\infty}_{s=0}\frac{1}{2s+1}\frac{k_j^{2r+1}}{k_i^{2s+1}}\delta_{r,s}\\
     &=-4 \sum^{\infty}_{r=0}\frac{1}{2r+1}\frac{k_j^{2r+1}}{k_i^{2r+1}}=- 4\zeta(\varepsilon(k_i),k_j).
\end{align*}
Then, using \eqref{AB-BA}, the relation \eqref{switch} can be written as
\begin{equation}
e^Ae^B=e^{[A,B]}e^Be^A.
\end{equation}

\begin{lem}\label{lem-2-6-4}
For the vertex operator $X(k)$ defined by \eqref{X-kdv-a} and $A_{ij}$ defined in \eqref{Nss-tau-b}, there is
\begin{equation}
X(k_i)X(k_j)=A_{ij}e^{2(\zeta(\mathbf{t},k_i)+\zeta(\mathbf{t},k_j))} e^{-2\zeta(\t \partial,\, k^{-1}_i)}  e^{-2\zeta(\t \partial,\, k^{-1}_j)}.
\label{XXkij}
\end{equation}
\end{lem}
\begin{proof}
Using formula \eqref{switch}, we find
\begin{align*}
X(k_i)X(k_j)& = e^{2\zeta(\mathbf{t},k_i)} e^{-2\zeta(\t \partial,\, k^{-1}_i)} e^{2\zeta(\mathbf{t},k_j))} e^{-2\zeta(\t \partial,\, k^{-1}_j)}\\
            & = e^{2\zeta(\mathbf{t},k_i)} A_{ij}\,e^{2\zeta(\mathbf{t},k_j))} e^{-2\zeta(\t \partial,\, k^{-1}_i)} e^{-2\zeta(\t \partial,\, k^{-1}_j)},
\end{align*}
which is \eqref{XXkij}.
\end{proof}

As a  corollary of Lemma \ref{lem-2-6-4}, we have the following.

\begin{lem}\label{lem-2-6-5}
For the vertex operator $X(k)$ defined by \eqref{X-kdv-a}, there are
\begin{align}
& (X(k))^2=0,\\
& e^{cX(k)}=1+cX(k),
\end{align}
\begin{align}
& X(k_s)\cdots X(k_2) X(k_1)\nonumber\\
=& \Biggl(\prod_{1\leq i<j}^s A_{ij}\Biggr)
\mathrm{exp}\Biggl(2\sum_{j=1}^s \zeta(\mathbf{t},k_j)\Biggr)
\mathrm{exp}\Biggl(2\sum_{j=1}^s \zeta(\t \partial,k_j^{-1})\Biggr);
\end{align}
and
\begin{align}
& X(k)\ccirc 1=e^{2\zeta(\mathbf{t},k)},\\
& X(k_s)\cdots X(k_2) X(k_1)\ccirc 1
= \Biggl(\prod_{1\leq i<j}^s A_{ij}\Biggr)
\mathrm{exp}\Biggl(2\sum_{j=1}^s \zeta(\mathbf{t},k_j)\Biggr),
\end{align}
where ``$\,{\ccirc} 1$" means an operator acting on $`` 1"$,
and $A_{ij}$ is defined in \eqref{Nss-tau-b}.
\end{lem}

Making use of Lemma \ref{lem-2-6-5}, it is not difficult to calculate:
\begin{align*}
\tau_1& =e^{c_1 X(k_1)}\ccirc 1=(1+c_1 X(k_1))\ccirc 1=1+c_1 e^{2\zeta(\mathbf{t},k_1)},\\
\tau_2& =e^{c_2 X(k_2)}e^{c_1 X(k_1)}\ccirc 1=(1+c_2 X(k_2))(1+c_1 X(k_1))\ccirc 1\\
      & =1+c_1 e^{2\zeta(\mathbf{t},k_1)}+c_2 e^{2\zeta(\mathbf{t},k_2)}+c_1c_2 A_{12} e^{2(\zeta(\mathbf{t},k_1)+\zeta(\mathbf{t},k_2))},\\
\tau_N& =e^{c_N X(k_N)}\cdots e^{c_2 X(k_2)} e^{c_1 X(k_1)}\ccirc 1\\
      & =(1+c_N X(k_N))\cdots (1+c_2 X(k_2))(1+c_1 X(k_1))\ccirc 1\\
      & =e^{c_N X(k_N)} \tau_{N-1}.
\end{align*}
Thus we also finish the proof for Theorem \ref{thm-2-6-1}.

\subsection{Vertex operator of the KP(II) equation}\label{sec-5-2}

Similar to \eqref{kdv-bil-D-tau-13}
we can rewrite the bilinear KP(II) equation \eqref{kp-bil-D} as the following,
\begin{equation}
(4D_1D_3 -D^4_1-3D^2_2) \tau\cdot \tau=0.
\label{kp-bil-D-tau-13}
\end{equation}
Its NSS \eqref{Nss-f-kp} is presented as
\begin{subequations}\label{Nss-tau-sato-kp}
\begin{equation}
\tau_N=\sum_{J\subset I}\Biggl[ \Biggl(\prod_{i\in J}c_i\Biggr)\Biggl(\prod_{\begin{smallmatrix}i,j\in J\\i<j\end{smallmatrix}}A_{ij}\Biggr)
\mathrm{exp}\Biggl(\sum_{i\in J} \xi_i\Biggr)\Biggr],
\label{Nss-tau-sato-kp-a}
\end{equation}
where $c_i\in \mathbb{R}$,
\begin{equation}
\xi_j=\sum_{i=0}^{\infty} (p_j^i-q_j^i)t_i, ~~ e^{a_{ij}}=A_{ij}=\frac{(p_i-p_j)(q_i-q_j)}{(p_i-q_j)(q_i-p_j)},
\label{Nss-tau-sato-kp-b}
\end{equation}
\end{subequations}
$I$ stands for the set $I=\{1,2,\cdots,N\}$, $J$ is a subset of $I$,
and summation over $J\subset I$ means taking all possible subsets of $I$.

For the vertex operator related to the KP(II) equation, we have the following \cite{DatKM-1981,MiwJD-2000}.

\begin{theorem}\label{thm-2-6-2}
For the $\tau$ function defined by \eqref{Nss-tau-sato-kp}, there is
\begin{equation}
\tau_{N+1}=e^{c_{N+1}X(p_{N+1},q_{N+1})} \tau_N,
\label{tau-X-kp}
\end{equation}
where
\begin{subequations}\label{X-kp}
\begin{align}
& X(p,q)=e^{\xi(\mathbf{t},p)-\xi(\mathbf{t},q)} \, e^{-(\xi(\t \partial,\, p^{-1})-\xi(\t \partial,\, q^{-1})},\label{X-kp-a}\\
& \xi(\mathbf{t},k)=\sum^{\infty}_{j=0}k^{j} t_{j},~~~ \mathbf{t}=(t_1, t_2, t_3, \cdots),\label{X-kp-b}\\
& \t{\partial}=\Bigl(\partial_{1},\, \frac{\partial_{2}}{2},\, \frac{\partial_{3}}{3},\, \cdots \Bigr),~~~ \partial_{j}=\partial_{t_j}.
\end{align}
\end{subequations}
\end{theorem}

\vskip 5pt

In fact, (similar to Sec.\ref{sec-5-1}) one can prove that
\begin{equation}
X(p_i,q_i)X(p_j,q_j)=A_{ij}\,:\!\!X(p_i,q_i)X(p_j,q_j)\!\!:,
\label{XXpqij}
\end{equation}
where   $A_{ij}$ is defined in \eqref{Nss-tau-sato-kp-b},
\[:\!\!X(p_i,q_i)X(p_j,q_j)\!\!:\,=e^{\xi(\mathbf{t},p_i)-\xi(\mathbf{t},q_i)} e^{\xi(\mathbf{t},p_j)-\xi(\mathbf{t},q_j)}
\, e^{-(\xi(\t \partial,\, p^{-1}_i)-\xi(\t \partial,\, q^{-1}_i))}
\, e^{-(\xi(\t \partial,\, p^{-1}_j)-\xi(\t \partial,\, q^{-1}_j))}\]
is the normally arranged product of $X(p_i,q_i)X(p_j,q_j)$, (just moving all differential operators to the most right).

Then we can have the following results which are similar to Lemma \ref{lem-2-6-5}:
\begin{align}
& (X(p,q))^2=0,\\
& e^{cX(p,q)}=1+cX(p,q),\\
& X(p_s,q_s)\cdots X(p_2,q_2) X(p_1,q_1)
=   \Biggl(\prod_{1\leq i<j}^s A_{ij}\Biggr)
:\!\!X(p_s,q_s)\cdots X(p_2,q_2) X(p_1,q_1)\!\! : \, ;
\end{align}
and
\begin{align}
& X(p,q)\ccirc 1=e^{\xi(\mathbf{t},p)-\xi(\mathbf{t},q)},\\
& X(p_s,q_s)\cdots X(p_2,q_2) X(p_1,q_1) \ccirc 1
=  \Biggl(\prod_{1\leq i<j}^s A_{ij}\Biggr)
\mathrm{exp}\Biggl(\sum_{j=1}^s (\xi(\mathbf{t},p_j)-\xi(\mathbf{t},q_j))\Biggr).
\end{align}
Then it is not difficult to get the results in Theorem \ref{thm-2-6-2}.

\section{Concluding remarks}\label{sec-6}

In this partial review of the bilinear method, we focused on two topics of the KdV-type bilinear equations:
the integrability  based on the 3SS condition
and the transformations between $\tau$ functions via BTs and vertex operators.
The existence of BTs also indicates integrability,
while the vertex operators reveal algebraic aspects of $\tau$ functions and the related integrable equations.

Physically, the elastic scattering property for arbitrary number of solitons
implies integrability.
For the KdV-type bilinear equation \eqref{kdvtype-bil}, which always admits a 1SS and a 2SS,
if it is integrable, then we can consider the 2SS as a result of a 3SS after removing one soliton
under the elastic scattering property.
It turns out that such a 3SS must take Hirota's form \eqref{kdv-type-Nss} for $N=3$.
Repeating this, one can increase the number of solitons  step by step, and finally recover the NSS \eqref{kdv-type-Nss}.
Along the above narration, it can be understood the existence of a 3SS
does mean something to the integrability of a KdV-type bilinear equation.
In Hietarinta's search for integrable bilinear equations,
except for the KdV-type \cite{Hie-1987a},
the mKdV-type and sG-type bilinear equations in his consideration
also admit 1SSs and   2SSs \cite{Hie-1987b,Hie-1987c}.
So, the existence of a 3SS might imply integrability for these types of bilinear equations.
For the NLS-type and BO-type bilinear equations considered in \cite{Hie-1988},
only 1SSs automatically exist, which means  the existence of a 2SS
may imply integrability.
However, strictly speaking, the 3SS condition (2SS for some bilinear equations, e.g.  the NLS-type)
is necessary but not universally sufficient for integrability for a general bilinear equation or a system,
and  additional tests (e.g. the Lax pair or Painlev\'e property) are often required.
Note also that recently the author and collaborators developed a procedure to
nonlinearize bilinear equations by using the connections between the Bell polynomials and
Hirota's bilinear derivatives \cite{ZLZ-CTP-2025,LZZZ-2026}.
This may provide a solution to the question raised by Hietarinta in \cite{Hie-1990}:
finding an algorithmic way to nonlinearize bilinear equations.

The BTs, on one side,
indicate a way to generate new solutions and Lax pairs
as well as integrability.
To construct a bilinear BT, except for by means of Hirota's technique as what we have shown in Sec.\ref{sec-4},
one can also employ the connections between the Bell polynomials and
Hirota's bilinear derivatives. Examples can be found in \cite{LamS-2006}.
On the other side, BTs provide a way to integrable discretization, which was  observed
by Levi and Benguria in 1980 \cite{LevB-1980}.
For the KdV equation, the nonlinear BT \eqref{kdv-BT-non-1-aa-1}
is nothing but the differential-difference potential KdV equation, cf.\cite{HieJN-2016},
and its bilinear form \eqref{kdv-BT-bil-x} is considered as a bilinear
differential-difference potential KdV equation.

For the vertex operators,
we only focus on the way they provide transformations for $\tau$ functions.
It is well known that once a vertex operator and its corresponding $\tau$ function are secured,
one can have a bilinear identity and then a hierarchy of bilinear equations
that have the $\tau$ function as their solution.
A recent progress is a building of the bilinear framework for elliptic solitons \cite{LZ-JNS-2022,LZ-CMAMS-2024},
of which the dispersion relations are non-polynomial type.
These solutions are obtained by replacing the classical exponential function by the Lam\'e function,
which describe solitons on an elliptic background.
Researchers are trying to understand more about elliptic solutions from Sato's theory for integrable systems,
e.g. \cite{Kak-arxiv-2023,Nak-2024,Nak-2025}.

\subsection*{Acknowledgements}

The author is very grateful to Jarmo Hietarinta for his guidance and collaboration
in the bilinear method and in the research of discrete integrable systems.
Sincere thanks are  extended to Prof. Ralph Willox for discussion,
to Prof. Yasuhiro Ohta for pointing out \cite{Hir-1977a}
as the reference to the formula \eqref{id-gen-1},
to Prof. Kenichi Maruno for sharing \cite{Hir-2004},
to Prof.  Daisuke Takahashi for sharing \cite{Tak-2016},
and to Prof. Yuji Kodama for sharing stories about the bilinear method and Sato's KP theory.
This review is completed partially based on the lecture notes the author gave in the University of Sydney in 2017.
He sincerely thank  Prof. Nalini Joshi for her hospitality during his visit in Sydney.
This project is supported by the NSFC grant (Nos. 1241540016, 12271334).

\label{lastpage}
\end{document}